\documentclass{article} 
\usepackage{nips13submit_e,times}
\usepackage{natbib}
  \bibpunct{(}{)}{;}{a}{}{,} 
\usepackage{pst-3dplot}

\usepackage[dvips]{graphicx}

\usepackage[tbtags]{amsmath}
\usepackage{amsbsy}
\usepackage{amssymb}
\usepackage{amsfonts}
\usepackage{amsthm}

\usepackage{graphicx}
\usepackage{url}
\usepackage{algorithmic}
\usepackage{algorithm}
\usepackage{balance}
\usepackage{rotating}
\usepackage{multirow}
\usepackage{subcaption}
\usepackage{fancyvrb}
\usepackage{latexsym}
\usepackage{verbatim}
\usepackage{hyperref}

\newtheorem{theorem}{Theorem}

\newtheorem{lemma}{Lemma}

\newtheorem{definition}{Definition}

{\begin{list}               
    {$\bullet$ \hfill}{
        \setlength{\leftmargin}{\parindent}
        \setlength{\parsep}{0.04\baselineskip}
        \setlength{\itemsep}{0.5\parsep}
        \setlength{\labelwidth}{\leftmargin}
        \setlength{\labelsep}{0em}}
    }
{\end{list}}

\providecommand{\eref}[1]{\eqref{#1}}  
\providecommand{\cref}[1]{Chapter~\ref{#1}}

\providecommand{\fref}[1]{Figure~\ref{#1}}

\providecommand{\thref}[1]{Theorem~\ref{#1}}

\providecommand{\E}{\ensuremath{\mathbb{E}}}

\providecommand{\calS}{\mathcal{S}}

\providecommand{\calO}{\mathcal{O}}

\providecommand{\calG}{\mathcal{G}}
\providecommand{\calE}{\mathcal{E}}

\providecommand{\wbar}{\overline{w}}
\providecommand{\what}{\widehat{w}}
\providecommand{\dhat}{\widehat{d}}
\providecommand{\rhat}{\widehat{r}}
\providecommand{\chat}{\widehat{c}}
\providecommand{\Bhat}{\widehat{B}}
\providecommand{\Vbar}{\overline{V}}

\providecommand{\Var}{\mathrm{Var}}
\providecommand{\Cov}{\mathrm{Cov}}

\newcommand{\defequal}{\mathop{\overset{\mbox{\tiny{def}}}{=}}}

\providecommand{\MSE}{\mathrm{MSE}}
\providecommand{\MAE}{\mathrm{MAE}}

\title{Stochastic blockmodel approximation of a graphon: Theory and consistent estimation\protect\thanks{This paper appears in the proceedings of NIPS 2013. In this version we include an appendix with proofs.}}

\author{
Edoardo M. Airoldi \\
Dept. Statistics\\
Harvard University
\And
Thiago B. Costa \\
SEAS, and Dept. Statistics\\
Harvard University
\And
Stanley H.~Chan \\
SEAS, and Dept. Statistics\\
Harvard University
}

\nipsfinalcopy 

\graphicspath{{pix/}}
\begin{document}

\maketitle

\begin{abstract}
Non-parametric approaches for analyzing network data based on exchangeable graph models (ExGM) have recently gained interest. The key object that defines an ExGM is often referred to as a {\em graphon}. This non-parametric perspective on network modeling poses challenging questions on how to make inference on the graphon underlying  observed network data. In this paper, we propose a computationally efficient procedure to estimate a graphon from a set of observed networks generated from it. This procedure is based on a stochastic blockmodel approximation (SBA) of the graphon. We show that, by approximating the graphon with a stochastic block model, the graphon can be consistently estimated, that is, the estimation error vanishes as the size of the graph approaches infinity.
\end{abstract}

\section{Introduction}
Revealing hidden structures of a graph is the heart of many data analysis problems. From the well-known small-world network to the recent large-scale data collected from online service providers such as Wikipedia, Twitter and Facebook, there is always a momentum in seeking better and more informative representations of the graphs \citep{Fien:Meye:Wass:1985,Nowi:Snij:2001,Hoff_Raftery_Handcock_2002,Hand:Raft:Tant:2007,Airoldi_Blei_Fienberg_2008,Xu_Yan_Qi_2012,Azari:2012fk,Tang_Sussman_Priebe_2013,Goldenberg_Zheng_Fienberg_2009,kola:2009}. In this paper, we develop a new computational tool to study one type of non-parametric representations which recently draws significant attentions from the community  \citep{Bickel_Chen_2009,Lloyd_Orbanz_Ghahramani_2012,bickel_2011,Zhao_2011,OR13}.

The root of the non-parametric model discussed in this paper is in the theory of exchangeable random arrays \citep{Aldous_1981,Hoover_1979,Kallenberg_1989}, and it is presented in \citep{Diaconis_Janson_2007} as a link connecting de Finetti's work on partial exchangeability and graph limits \citep{Lovasz_Szegedy_2006,Borgs_Chayes_Lovasz_2006}. In a nutshell, the theory predicts that every convergent sequence of graphs $(G_n)$ has a limit object that preserves many local and global properties of the graphs in the sequence. This limit object, which is called a \emph{graphon}, can be represented by measurable functions $w: [0,1]^2 \rightarrow [0,1]$, in a way that any $w'$ obtained from measure preserving transformations of $w$ describes the same graphon.

Graphons are usually seen as kernel functions for random network models \citep{Lawrence_2005}.  
To construct an $n$-vertex random graph $\calG(n,w)$ for a given $w$, we first assign a random label $u_i \sim \mbox{Uniform}[0,1]$ to each vertex $i \in \{1,\ldots,n\}$, and connect any two vertices $i$ and $j$ with probability $w(u_i,u_j)$, \emph{i.e.},
\begin{equation}
\Pr\left( G[i,j] = 1 \;|\; u_i, u_j\right) = w(u_i, u_j), \quad\quad i,j = 1,\ldots,n,
\label{eq:defineG}
\end{equation}
where $G[i,j]$ denotes the $(i,j)$th entry of the adjacency matrix representing a particular realization of $\calG(n,w)$ (See \fref{fig:illustration of graphon}). As an example, we note that the stochastic block-model is the case where $w(x,y)$ is a piecewise constant function.

\begin{figure}[t]
\centering
\includegraphics[width=\linewidth]{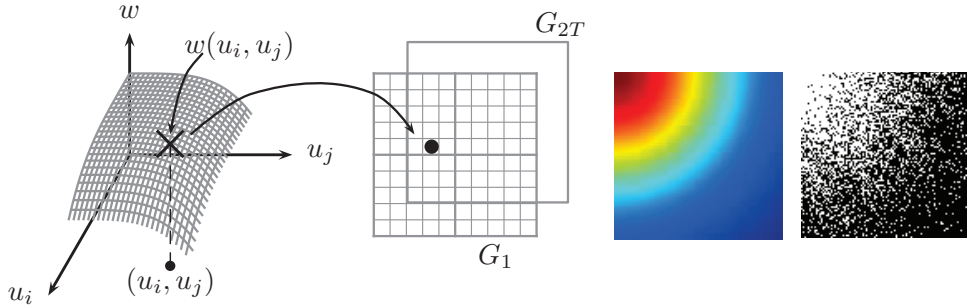}
\caption{[Left] Given a graphon $w: [0,1]^2 \rightarrow [0,1]$, we draw i.i.d. samples $u_i$, $u_j$ from $\mbox{Uniform[0,1]}$ and assign $G_t[i,j] = 1$ with probability $w(u_i,u_j)$, for $t = 1,\ldots,2T$. [Middle] Heat map of a graphon $w$. [Right] A random graph generated by the graphon shown in the middle. Rows and columns of the graph are ordered by increasing $u_i$, instead of $i$ for better visualization.}
\label{fig:illustration of graphon}
\end{figure}

The problem of interest is  defined as follows: Given a sequence of $2T$ observed \emph{directed} graphs $G_1, \ldots, G_{2T}$, can we make an estimate $\what$ of $w$, such that $\what \rightarrow w$ with high probability as $n$ goes to infinity? This question has been loosely attempted in the literature, but none of which has a complete solution. For example, Lloyd et al. \citep{Lloyd_Orbanz_Ghahramani_2012} proposed a Bayesian estimator without a consistency proof; Choi and Wolfe \citep{Choi_Wolfe_2012} studied the consistency properties, but did not provide algorithms to estimate the graphon. To the best of our knowledge, the only method that estimates graphons consistently, besides ours, is USVT \citep{Chatterjee2012}. However, our algorithm has better complexity and outperforms USVT in our simulations. 
More recently, other groups have begun exploring approaches related to ours  \citep{Wolfe_Olhede_2013,Latouche_Robin_2013}.

The proposed approximation procedure  requires $w$ to be piecewise Lipschitz. The basic idea is to approximate $w$ by a two-dimensional step function $\what$ with diminishing intervals as $n$ increases.The proposed method is called the Stochastic blockmodel approximation (SBA) algorithm, as the idea of using a two-dimensional step function for approximation is equivalent to using the stochastic block models \citep{Choi_Wolfe_Airoldi_2012,Nowi:Snij:2001,Hoff_2008,Channarond_Daudin_Robin_2012,rohe_2011}. The SBA algorithm is defined up to permutations of the nodes, so the estimated graphon is \emph{not} canonical. However, this does not affect the consistency properties of the SBA algorithm, as the consistency is measured w.r.t. the graphon that generates the graphs.

\section{Stochastic blockmodel approximation: Procedure}
\label{sec:proposed method}
In this section we present the proposed SBA algorithm and discuss its basic properties.

\subsection{Assumptions on graphons}
We assume that $w$ is \emph{piecewise Lipschitz}, \emph{i.e.}, there exists a sequence of non-overlaping intervals $I_k = [\alpha_{k-1},\,\alpha_k]$ defined by
$0 = \alpha_0 < \ldots < \alpha_K = 1$, and a constant $L>0$ such that, for any $(x_1,y_1)$ and $(x_2,y_2) \in I_{ij} = I_i \times I_j$,
\begin{equation*}
|w(x_1,y_1) - w(x_2,y_2)| \le L \left(|x_1 - x_2| + |y_1 - y_2|\right).
\end{equation*}
For generality we assume $w$ to be asymmetric \emph{i.e.}, $w(u,v) \not= w(v,u)$, so that symmetric graphons can be considered as a special case. Consequently, a random graph $\calG(n,w)$ generated by $w$ is directed, \emph{i.e.}, $G[i,j] \not= G[j,i]$.

\subsection{Similarity of graphon slices}
The intuition of the proposed SBA algorithm is that if the graphon is smooth, neighboring cross-sections of the graphon should be similar. In other words, if two labels $u_i$ and $u_j$ are close \emph{i.e.}, $|u_i - u_j| \approx 0$, then the difference between the row slices $|w(u_i,\cdot)-w(u_j,\cdot)|$ and the column slices $|w(\cdot,u_i)-w(\cdot,u_j)|$ should also be small.
To measure the similarity between two labels using the graphon slices, we define the following distance
\begin{equation}
d_{ij} = \frac{1}{2} \left(\int_0^1 \left[w(x,u_i)-w(x,u_j)\right]^2 dx + \int_0^1 \left[w(u_i,y)-w(u_j,y)\right]^2 dy\right).
\end{equation}
Thus, $d_{ij}$ is small only if both row and column slices of the graphon are similar.

The usage of $d_{ij}$ for graphon estimation will be discussed in the next subsection. But before we proceed, it should be noted that in practice $d_{ij}$ has to be estimated from the observed graphs $G_1, \ldots, G_{2T}$. To derive an estimator $\dhat_{ij}$ of $d_{ij}$, it is helpful to express $d_{ij}$ in a way that the estimators can be easily obtained. To this end, we let
\begin{align*}
c_{ij} =  \int_0^1 w(x,u_i)w(x,u_j)dx \quad\quad\mbox{and}\quad\quad r_{ij} = \int_0^1w(u_i,y)w(u_j,y)dy,
\end{align*}
and express $d_{ij}$ as $d_{ij} = \frac{1}{2}\Big[ (c_{ii} - c_{ij} - c_{ji} + c_{jj}) + (r_{ii} - r_{ij} - r_{ji} + r_{jj})\Big]$. Inspecting this expression, we consider the following estimators for $c_{ij}$ and $r_{ij}$:
\begin{align}
\chat_{ij}^k &= \frac{1}{T^2}\left( \sum_{1\le t_1 \le T}G_{t_1}[k,i]\right)\left( \sum_{T < t_2 \le 2T}G_{t_2}[k,j]\right), \label{eq:chat}\\
\rhat_{ij}^k &= \frac{1}{T^2}\left( \sum_{1\le t_1 \le T}G_{t_1}[i,k]\right)\left( \sum_{T < t_2 \le 2T}G_{t_2}[j,k]\right). \label{eq:rhat}
\end{align}
Here, the superscript $k$ can be interpreted as the dummy variables $x$ and $y$ in defining $c_{ij}$ and $r_{ij}$, respectively. Summing all possible $k$'s yields an estimator $\dhat_{ij}$ that looks similar to $d_{ij}$:
\begin{equation}
\dhat_{ij} = \frac{1}{2} \left[ \frac{1}{S} \sum_{k \in \calS} \left\{   \left(\rhat_{ii}^k - \rhat_{ij}^k - \rhat_{ji}^k + \rhat_{jj}^k\right) + \left(\chat_{ii}^k - \chat_{ij}^k - \chat_{ji}^k + \chat_{jj}^k\right)  \right\} \right],
\end{equation}
where $\calS = \{1,\ldots,n\}\backslash\{i,j\}$ is the set of summation indices.

The motivation of defining the estimators in \eref{eq:chat} and \eref{eq:rhat} is that a row of the adjacency matrix $G[i,\cdot]$ is fully characterized by the corresponding row of the graphon $w(u_i,\cdot)$. Thus the expected value of $\frac{1}{T}\left( \sum_{1\le t_1 \le T}G_{t_1}[i,\cdot]\right)$ is  $w(u_i,\cdot)$, and hence $\frac{1}{S}\sum_{k \in \calS} \rhat_{ij}^k$ is an estimator for $r_{ij}$. To theoretically justify this intuition, we will show in Section \ref{sec:theory} that $\dhat_{ij}$ is indeed a good estimator: it is not only unbiased, but is also concentrated round $d_{ij}$ for large $n$. Furthermore, we will show that it is possible to use a random subset of $\calS$ instead of $\{1,\ldots,n\}\backslash\{i,j\}$ to achieve the same asymptotic behavior. As a result, the estimation of $d_{ij}$ can be performed locally in a neighborhood of $i$ and $j$, instead of all $n$ vertices.

\subsection{Blocking the vertices}
The similarity metric $\dhat_{ij}$ discussed above suggests one simple method to approximate $w$ by a piecewise constant function $\what$ (\emph{i.e.}, a stochastic block-model). Given $G_{1},\ldots,G_{2T}$, we can cluster the (unknown) labels $\{u_1,\ldots,u_n\}$ into $K$ blocks $\Bhat_1,\ldots,\Bhat_K$ using a procedure described below. Once the blocks $\Bhat_1,\ldots,\Bhat_K$ are defined, we can then determine $\what(u_i,u_j)$ by computing the empirical frequency of edges that are present across blocks $\Bhat_i$ and $\Bhat_j$:
\begin{equation}
\what(u_i,u_j) = \frac{1}{|\Bhat_i|\,|\Bhat_j|} \sum_{i_x \in \Bhat_i} \sum_{j_y \in \Bhat_j} \frac{1}{2T} \left( G_1[i_x,j_y] + G_2[i_x,j_y] + \ldots + G_{2T}[i_x,j_y]\right),
\label{eq:what}
\end{equation}
where $\Bhat_i$ is the block containing $u_i$ so that summing $G_t[x,y]$ over $x \in \Bhat_i$ and $y \in \Bhat_j$ yields an estimate of the expected number of edges linking block $\Bhat_i$ and $\Bhat_j$.

To cluster the unknown labels $\{u_1,\ldots,u_n\}$ we propose a greedy approach as shown in Algorithm \ref{alg:SBA}. Starting with $\Omega = \{u_1,\ldots,u_n\}$, we randomly pick a node $i_p$ and call it the \emph{pivot}. Then for all other vertices $i_v \in \Omega\backslash\{i_p\}$, we compute the distance $\dhat_{i_p,i_v}$ and check whether $\dhat_{i_p,i_v} < \Delta^2$ for some precision parameter $\Delta>0$. If $\dhat_{i_p,i_v} < \Delta^2$, then we assign $i_v$ to the same block as $i_p$. Therefore, after scanning through $\Omega$ once, a block $\Bhat_1 = \{i_p, i_{v_1}, i_{v_2}, \ldots\}$ will be defined. By updating $\Omega$ as $\Omega \leftarrow \Omega\backslash \Bhat_1$, the process repeats until $\Omega = \emptyset$.

The proposed greedy algorithm is only a local solution in a sense that it does not return the globally optimal clusters. However, as will be shown in Section \ref{sec:theory}, although the clustering algorithm is not globally optimal, the estimated graphon $\what$ is still guaranteed to be a consistent estimate of the true graphon $w$ as $n \rightarrow \infty$. Since the greedy algorithm is numerically efficient, it serves as a practical computational tool to estimate $w$.

\subsection{Main algorithm}
\begin{algorithm}[ht]
\centering
\caption{Stochastic blockmodel approximation}
\begin{algorithmic}
\STATE Input: A set of observed graphs $G_1, \ldots, G_{2T}$ and the precision parameter $\Delta$.
\STATE Output: Estimated stochastic blocks $\Bhat_1, \ldots, \Bhat_K$.
\STATE Initialize: $\Omega = \{1,\ldots,n\}$, and $k = 1$.
\WHILE{$\Omega \not= \emptyset$}
    \STATE Randomly choose a vertex $i_p$ from $\Omega$ and assign it as the pivot for $\Bhat_k$: $\Bhat_k \leftarrow i_p$.
    \FOR {Every other vertices $i_v \in \Omega \backslash \{i_p\}$}
        \STATE Compute the distance estimate $\dhat_{i_p, i_v}$.
        \STATE If $\dhat_{i_p, i_v} \le \Delta^2$, then assign $i_v$ as a member of $\Bhat_k$: $\Bhat_k \leftarrow i_v$.
    \ENDFOR
    \STATE Update $\Omega$: $\Omega \leftarrow \Omega \backslash \Bhat_k$.
    \STATE Update counter: $k \leftarrow k+1$.
\ENDWHILE
\end{algorithmic}
\label{alg:SBA}
\end{algorithm}
Algorithm \ref{alg:SBA} illustrates the pseudo-code for the proposed stochastic block-model approximation. The complexity of this algorithm is $\calO( T S K n)$, where $T$ is half the number of observations, $S$ is the size of the neighborhood, $K$ is the number of blocks and $n$ is number of vertices of the graph.

\section{Stochastic blockmodel approximation: Theory of estimation}
\label{sec:theory}

In this section we present the theoretical aspects of the proposed SBA algorithm. We will first discuss the properties of the estimator $\dhat_{ij}$, and then show the consistency of the estimated graphon $\what$. Details of the proofs can be found in the supplementary material.

\subsection{Concentration analysis of $\dhat_{ij}$}
Our first theorem below shows that the proposed estimator $\dhat_{ij}$ is both unbiased, and is concentrated around its expected value $d_{ij}$.
\begin{theorem}
\label{thm:dhatij}
The estimator $\dhat_{ij}$ for $d_{ij}$ is unbiased, i.e., $\E[\dhat_{ij}] = d_{ij}$. Further, for any $\epsilon > 0$,
\begin{equation}
\Pr\left[ \left|\dhat_{ij}  - d_{ij}\right| > \epsilon \right] \le 8e^{-\frac{S\epsilon^2}{32/T+8\epsilon/3}},
\end{equation}
where $S$ is the size of the neighborhood $\calS$, and $2T$ is the number of observations.
\end{theorem}

\begin{proof}
Here we only highlight the important steps to present the intuition. The basic idea of the proof is to zoom-in a microscopic term of $\rhat_{ij}^k$ and show that it is unbiased. To this end, we use the fact that $G_{t_1}[i,k]$ and $G_{t_2}[j,k]$ are conditionally independent on $u_k$ to show
\begin{align*}
\E[ G_{t_1}[i,k]G_{t_2}[j,k] \;|\; u_k] &= \Pr[G_{t_1}[i,k] = 1, G_{t_2}[j,k] = 1 \;|\; u_k]\\
&\overset{(a)}{=} \Pr[G_{t_1}[i,k] = 1 \;|\; u_k] \Pr[G_{t_2}[j,k] = 1 \;|\; u_k]\\
&= w(u_i,u_k)w(u_j,u_k),
\end{align*}
which then implies $\E[\rhat_{ij}^k \;|\; u_k] = w(u_i,u_k)w(u_j,u_k)$, and by iterated expectation we have $\E[\rhat_{ij}^k] = \E[ \E[\rhat_{ij}^k \;|\; u_k] ] = r_{ij}$. The concentration inequality follows from a similar idea to bound the variance of $\rhat_{ij}^k$ and apply Bernstein's inequality.
\end{proof}

That $G_{t_1}[i,k]$ and $G_{t_2}[j,k]$ are conditionally independent on $u_k$ is a critical fact for the success of the proposed algorithm. It also explains why at least 2 independently observed graphs are necessary, for otherwise we cannot separate the probability in the second equality above marked with $(a)$.

\subsection{Choosing the number of blocks}
The performance of the Algorithm \ref{alg:SBA} is sensitive to the number of blocks it defines. On the one hand, it is desirable to have more blocks so that the graphon can be finely approximated. But on the other hand, if the number of blocks is too large then each block will contain only few vertices. This is bad because in order to estimate the value on each block, a sufficient number of vertices in each block is required. The trade-off between these two cases is controlled by the precision parameter $\Delta$: a large $\Delta$ generates few large clusters, while small $\Delta$ generates many small clusters. A precise relationship between the $\Delta$ and $K$, the number of blocks generated the algorithm, is given in \thref{thm:number of blocks}.

\begin{theorem}
\label{thm:number of blocks}
Let $\Delta$ be the accuracy parameter and $K$ be the number of blocks estimated by Algorithm \ref{alg:SBA}, then
\begin{equation}
\Pr\left[ K > \frac{QL\sqrt{2}}{\Delta} \right] \le 8n^2 e^{-\frac{S\Delta^4}{128/T + 16\Delta^2/3}},
\end{equation}
where $L$ is the Lipschitz constant and $Q$ is the number of Lipschitz blocks in $w$.
\end{theorem}

In practice, we estimate $\Delta$ using a cross-validation scheme to find the optimal 2D histogram bin width \citep{Wasserman_2005}. The idea is to test a sequence of potential values of $\Delta$ and seek the one that minimizes the cross validation risk, defined as
\begin{equation}
\widehat{J}(\Delta)=\frac{2}{h(n-1)} - \frac{n+1}{h(n-1)}\sum_{j=1}^K \widehat{p}_j^2,
\end{equation}
where $\widehat{p}_j = |\Bhat_j|/n$ and $h = 1/K$. 
Algorithm \ref{alg:CrossValidation} details the proposed cross-validation scheme.

\begin{algorithm}
\caption{Cross Validation}
\begin{algorithmic}
\STATE Input: Graphs $G_1,\ldots,G_{2T}$.
\STATE Output: Blocks $\Bhat_1,\ldots,\Bhat_K$, and optimal $\Delta$.
\FOR{a sequence of $\Delta$'s}
    \STATE Estimate blocks $\Bhat_1,\ldots,\Bhat_K$ from $G_1,\ldots,G_{2T}$. [Algorithm \ref{alg:SBA}]
    \STATE Compute $\widehat{p}_j = |\Bhat_j|/n$, for $j=1,\ldots,K$.
    \STATE Compute $\widehat{J}(\Delta)=\frac{2}{h(n-1)} - \frac{n+1}{h(n-1)}\sum_{j=1}^K \widehat{p}_j^2$, with $h = 1/K$.
\ENDFOR
\STATE Pick the $\Delta$ with minimum $\widehat{J}(\Delta)$, and the corresponding $\Bhat_1,\ldots,\Bhat_K$.
\end{algorithmic}
\label{alg:CrossValidation}
\end{algorithm}

\subsection{Consistency of $\what$}
The goal of our next theorem is to show that $\what$ is a consistent estimate of $w$, \emph{i.e.}, $\what \rightarrow w$ as $n \rightarrow \infty$. To begin with, let us first recall two commonly used metric:
\begin{definition}
\label{def:MSE and MAE}
The mean squared error (MSE) and mean absolute error (MAE) are defined as
\begin{align*}
\MSE(\what) &= \frac{1}{n^2} \sum_{i_v = 1}^n \sum_{j_v = 1}^n \left( w(u_{i_v}, u_{j_v}) - \what(u_{i_v}, u_{j_v}) \right)^2\\
\MAE(\what) &= \frac{1}{n^2} \sum_{i_v = 1}^n \sum_{j_v = 1}^n \left| w(u_{i_v}, u_{j_v}) - \what(u_{i_v}, u_{j_v}) \right|.
\end{align*}
\end{definition}

\begin{theorem}
\label{thm:consistency}
If $S \in \Theta(n)$ and $\Delta \in \omega \left( \left(\frac{\log(n)}{n}\right)^{\frac{1}{4}} \right) \cap o(1)$, then
\begin{equation*}
\lim_{n\rightarrow \infty} \E[ \MAE(\what)] = 0 \quad\quad\mbox{and}\quad\quad \lim_{n\rightarrow \infty} \E[ \MSE(\what)] = 0.
\end{equation*}
\end{theorem}

\begin{proof}
The details of the proof can be found in the supplementary material . Here we only outline the key steps to present the intuition of the theorem. The goal of \thref{thm:consistency} is to show convergence of $|\what(u_i,u_j)-w(u_i,u_j)|$. The idea is to consider the following two quantities:
\begin{align*}
\wbar(u_i,u_j) &= \frac{1}{|\Bhat_i|\,|\Bhat_j|} \sum_{i_x \in \Bhat_i} \sum_{j_x \in \Bhat_j} w(u_{i_x}, u_{j_x}),\\
\what(u_i,u_j) &= \frac{1}{|\Bhat_i|\,|\Bhat_j|} \sum_{i_x \in \Bhat_i} \sum_{j_y \in \Bhat_j} \frac{1}{2T} \left( G_1[i_x,j_y] + G_2[i_x,j_y] + \ldots + G_{2T}[i_x,j_y]\right),
\end{align*}
so that if we can bound $|\wbar(u_i,u_j)-w(u_i,u_j)|$ and $|\wbar(u_i,u_j)-\what(u_i,u_j)|$, then consequently $|\what(u_i,u_j)-w(u_i,u_j)|$ can also be bounded.

The bound for the first term $|\wbar(u_i,u_j)-w(u_i,u_j)|$ is shown in Lemma \ref{lemma:bound wbar and w}: By Algorithm \ref{alg:SBA}, any vertex $i_v \in \Bhat_i$ is guaranteed to be within a distance $\Delta$ from the pivot of $\Bhat_i$. Since $\wbar(u_i,u_j)$ is an average over $\Bhat_i$ and $\Bhat_j$, by \thref{thm:dhatij} a probability bound involving $\Delta$ can be obtained.

The bound for the second term $|\wbar(u_i,u_j)-\what(u_i,u_j)|$ is shown in Lemma \ref{lemma:bound what and wbar}. Different from Lemma \ref{lemma:bound wbar and w}, here we need to consider two possible situations: either the intermediate estimate $\wbar(u_i,u_j)$ is close to the ground truth $w(u_i,u_j)$, or $\wbar(u_i,u_j)$ is far from the ground truth $w(u_i,u_j)$. This accounts for the sum in Lemma \ref{lemma:bound what and wbar}. Individual bounds are derived based on Lemma \ref{lemma:bound wbar and w} and \thref{thm:dhatij}.

Combining Lemma \ref{lemma:bound wbar and w} and Lemma \ref{lemma:bound what and wbar}, we can then bound the error and show convergence.
\end{proof}

\begin{lemma}
\label{lemma:bound wbar and w}
For any $i_v \in \Bhat_i$ and $j_v \in \Bhat_j$,
\begin{equation}
\Pr\left[ |\wbar(u_i,u_j) - w(u_{i_v},u_{j_v})| > 8\Delta^{1/2}L^{1/4} \right] \le 32|\Bhat_i|\,|\Bhat_j| e^{-\frac{S\Delta^4}{32/T+8\Delta^2/3}}.
\label{eq:bound wbar and w}
\end{equation}
\end{lemma}

\begin{lemma}
\label{lemma:bound what and wbar}
For any $i_v \in \Bhat_i$ and $j_v \in \Bhat_j$,
\begin{equation}
\Pr\left[ |\what_{ij} - \wbar_{ij}| > 8\Delta^{1/2}L^{1/4} \right]  \le 2 e^{-256 (T |\Bhat_i|\,|\Bhat_j| \sqrt{L} \Delta)} + 32 |\Bhat_i|^2 |\Bhat_j|^2 e^{-\frac{S\Delta^4}{32/T+8\Delta^2/3)}}.
\label{eq:bound what and wbar}
\end{equation}
\end{lemma}

The condition $S \in \Theta(n)$ is necessary to make \thref{thm:consistency} valid, because if $S$ is independent of $n$, it is not possible to drive \eref{eq:bound wbar and w} and \eref{eq:bound what and wbar} to $0$ even if $n \rightarrow \infty$. The other condition on $\Delta$ is also important as it forces the numerators and denominators in the exponentials of \eref{eq:bound wbar and w} and \eref{eq:bound what and wbar} to be well behaved.

\section{Experiments}
\label{sec:experiments}
In this section we evaluate the proposed SBA algorithm by showing some empirical results. For the purpose of comparison, we consider (i) the universal singular value thresholding (USVT) \citep{Chatterjee2012}; (ii) the largest-gap algorithm (LG) \citep{Channarond_Daudin_Robin_2012}; (iii) matrix completion from few entries (OptSpace) \citep{Keshavan_Montanari_Oh_2010}.

\subsection{Estimating stochastic blockmodels}
\paragraph{Accuracy as a function of growing graph size.}
Our first experiment is to evaluate the proposed SBA algorithm for estimating stochastic blockmodels. For this purpose, we generate (arbitrarily) a graphon
\begin{equation}
w = \begin{bmatrix}
     0.8 &0.9  &0.4 & 0.5\\
     0.1 &0.6  &0.3 & 0.2\\
     0.3 &0.2  &0.8 & 0.3\\
     0.4 &0.1  &0.2 & 0.9
     \end{bmatrix},
\label{eq:graphon example}
\end{equation}
which represents a piecewise constant function with $4 \times 4$ equi-space blocks.

\begin{figure}[h]
\centering
\begin{tabular}{cc}
\includegraphics[width=0.5\linewidth]{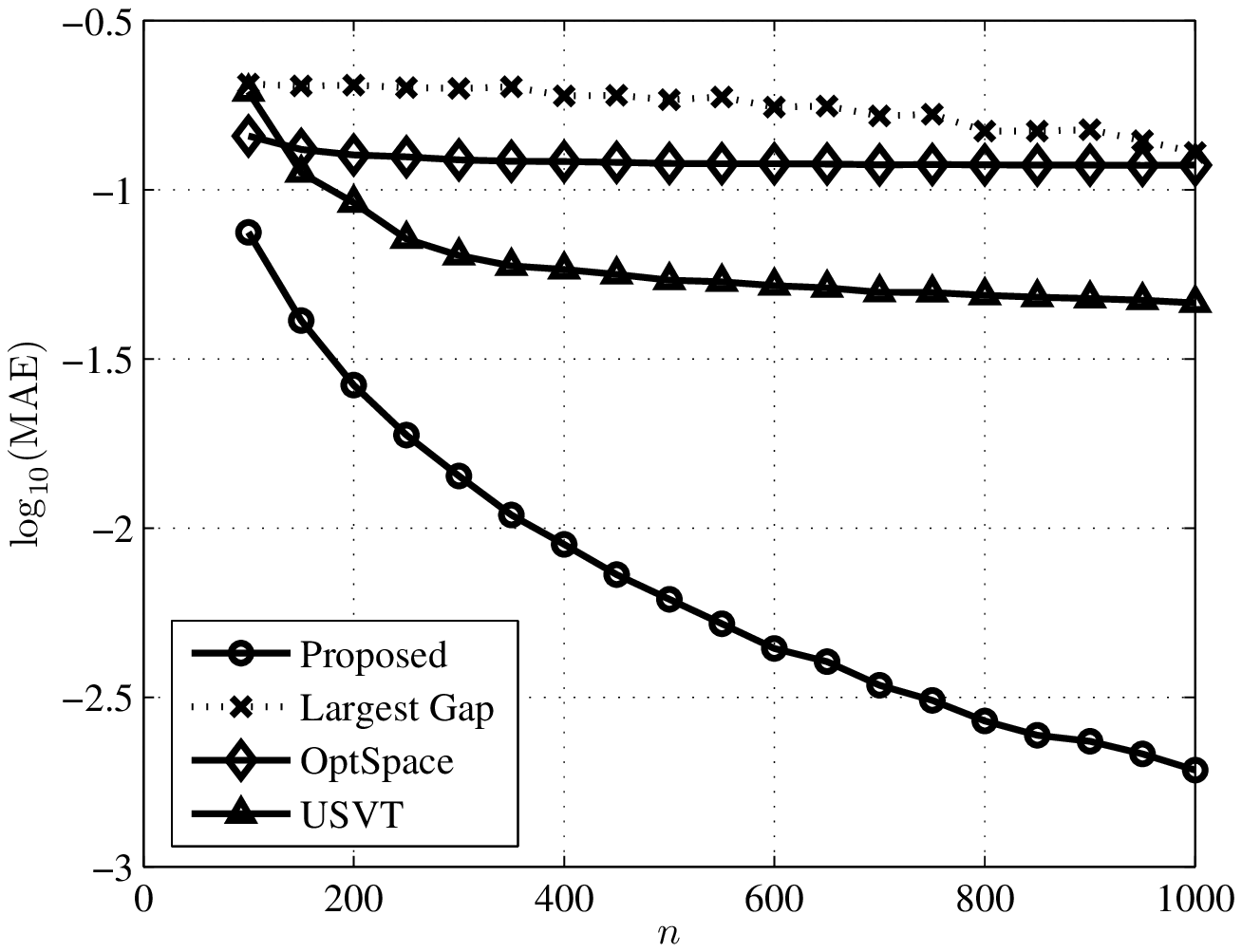}&
\includegraphics[width=0.5\linewidth]{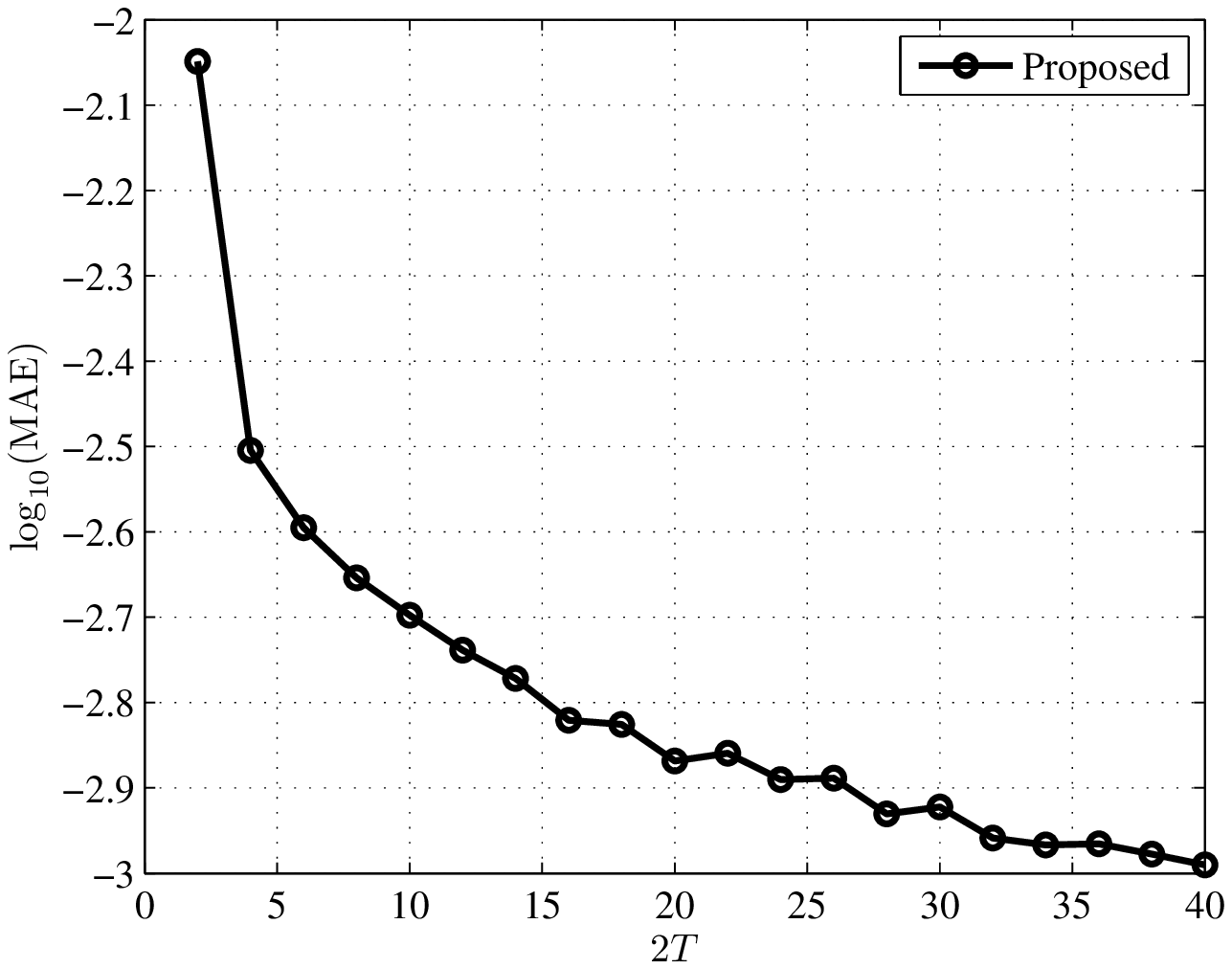}\\
(a) Growing graph size, $n$  & (b) Growing no. observations, $2T$
\end{tabular}
\caption{(a) MAE reduces as graph size grows. For the fairness of the amount of data that can be used, we use $\frac{n}{2} \times \frac{n}{2} \times 2$ observations for SBA, and $n \times n \times 1$ observation for USVT \citep{Chatterjee2012} and LG \citep{Channarond_Daudin_Robin_2012}. (b) MAE of the proposed SBA algorithm reduces when more observations $T$ is available. Both plots are averaged over 100 independent trials.}
\label{fig:Plot1 and Plot2}
\end{figure}

Since USVT and LG use only one observed graph whereas the proposed SBA require at least $2$ observations, in order to make the comparison fair, we use half of the nodes for SBA by generating two independent $\frac{n}{2} \times \frac{n}{2}$ observed graphs. For USVT and LG, we use one $n \times n$ observed graph.

\fref{fig:Plot1 and Plot2}(a) shows the asymptotic behavior of the algorithms when $n$ grows. \fref{fig:Plot1 and Plot2}(b) shows the estimation error of SBA algorithm as $T$ grows for graphs of size 200 vertices.


\paragraph{Accuracy as a function of growing number of blocks.}
Our second experiment is to evaluate the performance of the algorithms as $K$, the number of blocks, increases. To this end, we consider a sequence of $K$, and for each $K$ we generate a graphon $w$ of $K \times K$ blocks. Each entry of the block is a random number generated from $\mbox{Uniform}[0,1]$. Same as the previous experiment, we fix $n = 200$ and $T = 1$. The experiment is repeated over 100 trials so that in every trial a different graphon is generated. The result shown in \fref{fig:Plot3 and Plot4}(a) indicates that while estimation error increases as $K$ grows, the proposed SBA algorithm still attains the lowest MAE for all $K$.

\begin{figure}[ht]
\centering
\begin{tabular}{cc}
\includegraphics[width=0.5\linewidth]{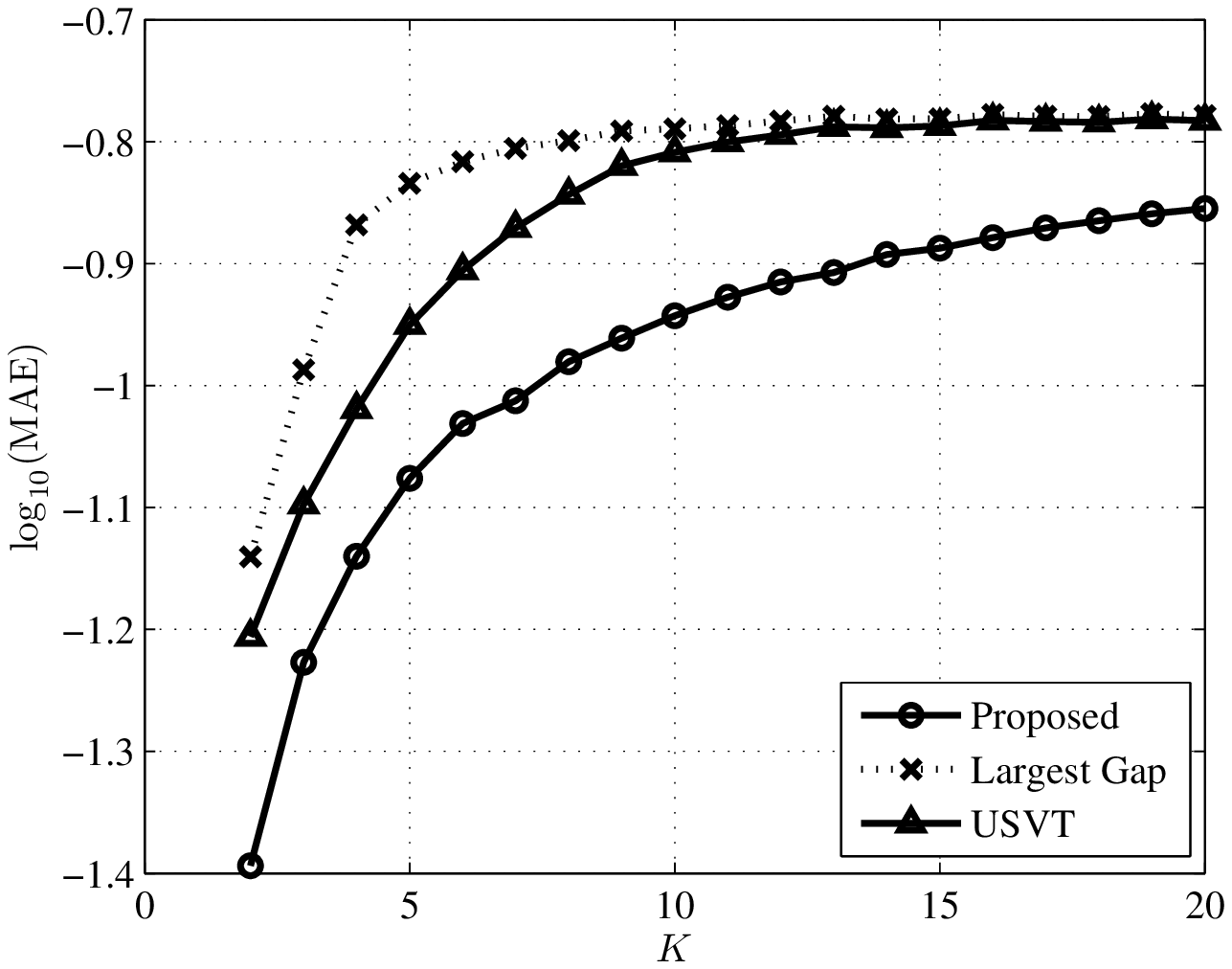}&
\includegraphics[width=0.5\linewidth]{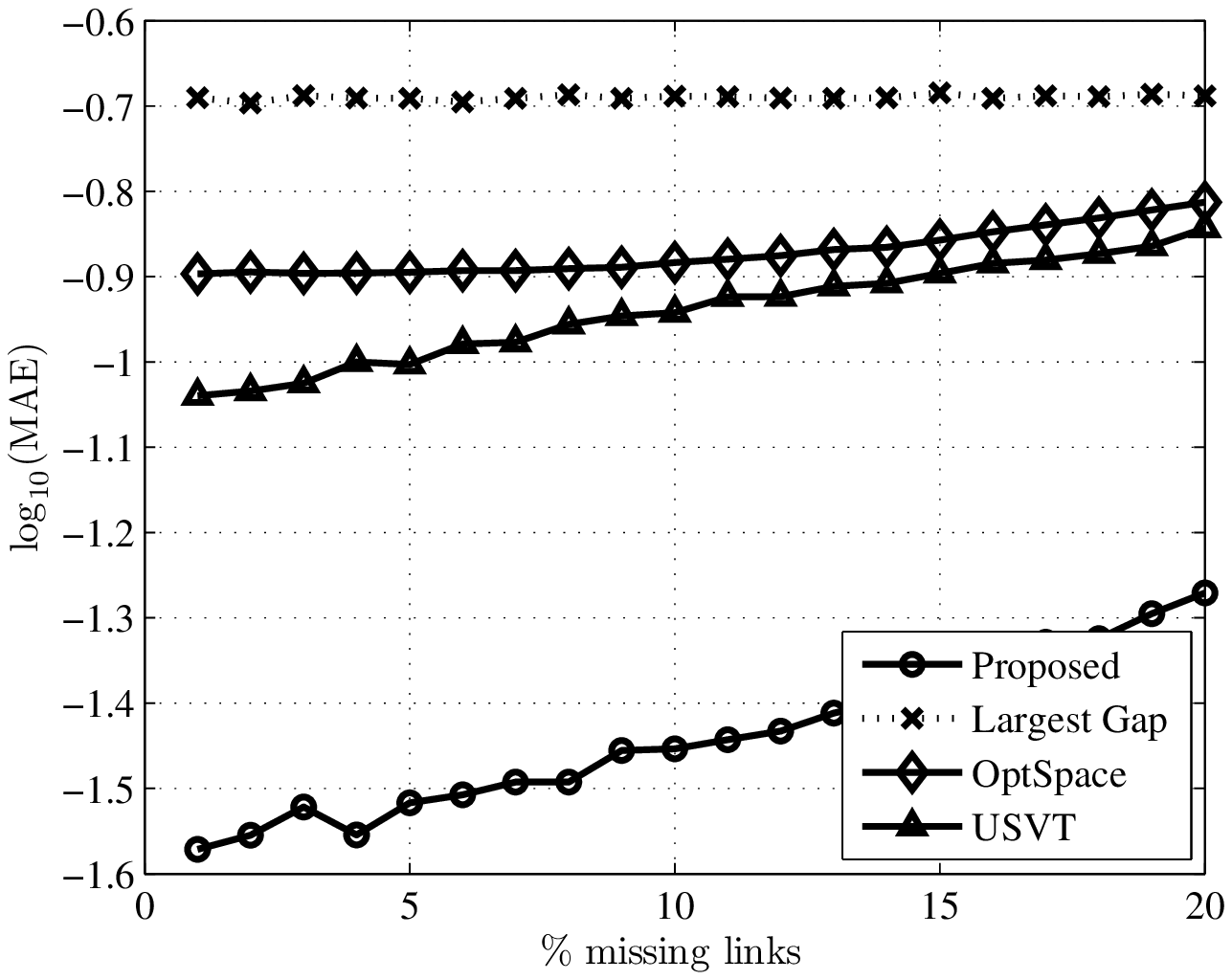}\\
(a) Growing no. blocks, $K$ & (b) Missing links
\end{tabular}
\caption{(a) As $K$ increases, MAE of all three algorithm increases but SBA still attains the lowest MAE. Here, we use $\frac{n}{2} \times \frac{n}{2} \times 2$ observations for SBA, and $n \times n \times 1$ observation for USVT \citep{Chatterjee2012} and LG \citep{Channarond_Daudin_Robin_2012}. (b) Estimation of graphon in the presence of missing links: As the amount of missing links increases, estimation error also increases.}
\label{fig:Plot3 and Plot4}
\end{figure}

\subsection{Estimation with missing edges}
Our next experiment is to evaluate the performance of proposed SBA algorithm when there are missing edges in the observed graph. To model missing edges, we construct an $n \times n$ binary matrix $M$ with probability $\Pr[M[i,j] = 0] = \xi$, where $0 \le \xi \le 1$ defines the percentage of missing edges. Given $\xi$,  $2T$  matrices are generated with missing edges, and the observed graphs are defined as $M_1 \odot G_1,\ldots,M_{2T} \odot G_{2T}$, where $\odot$ denotes the element-wise multiplication. The goal is to study how well SBA can reconstruct the graphon $\what$ in the presence of missing links.

The modification of the proposed SBA algorithm for the case missing links is minimal: when computing \eref{eq:what}, instead of averaging over all $i_x \in \Bhat_i$ and $j_y \in \Bhat_j$, we only average $i_x \in \Bhat_i$ and $j_y \in \Bhat_j$ that are not masked out by all $M'$s. \fref{fig:Plot3 and Plot4}(b) shows the result of average over 100 independent trials. Here, we consider the graphon given in \eref{eq:graphon example}, with $n=200$ and $T = 1$. It is evident that SBA outperforms its counterparts at a lower rate of missing links.

\subsection{Estimating continuous graphons}
Our final experiment is to evaluate the proposed SBA algorithm in estimating continuous graphons. Here, we consider two of the graphons reported in \citep{Chatterjee2012}:
\begin{align*}
w_1(u,v) = \frac{1}{1+\exp\{-50(u^2+v^2)\}}, \quad\mbox{and}\quad w_2(u,v) = uv,
\end{align*}
where $u, v \in [0,1]$. Here, $w_2$ can be considered as a special case of the Eigenmodel \citep{Hoff_2008} or latent feature relational model \citep{Miller_Griffiths_Jordan_2009}.

The results in \fref{fig:Plot5 and Plot6} shows that while both algorithms have improved estimates when $n$ grows, the performance depends on which of $w_1$ and $w_2$ that we are studying. This suggests that in practice the choice of the algorithm should depend on the expected structure of the graphon to be estimated: If the graph generated by the graphon demonstrates some low-rank properties, then  USVT is likely to be a better option. For more structured or complex graphons the proposed  procedure is recommended.

\begin{figure}[ht]
\centering
\begin{tabular}{cc}
\includegraphics[width=0.5\linewidth]{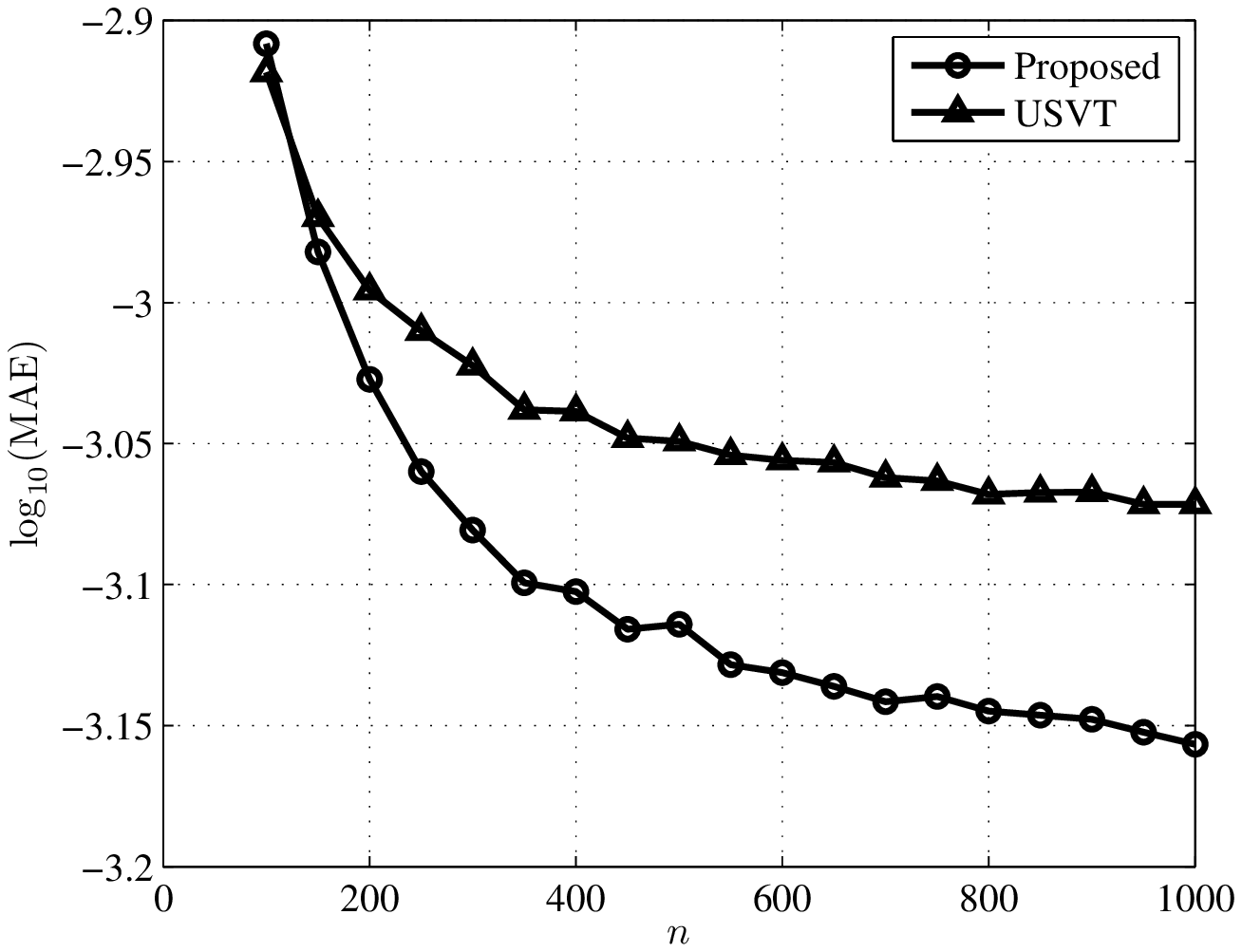}&
\includegraphics[width=0.5\linewidth]{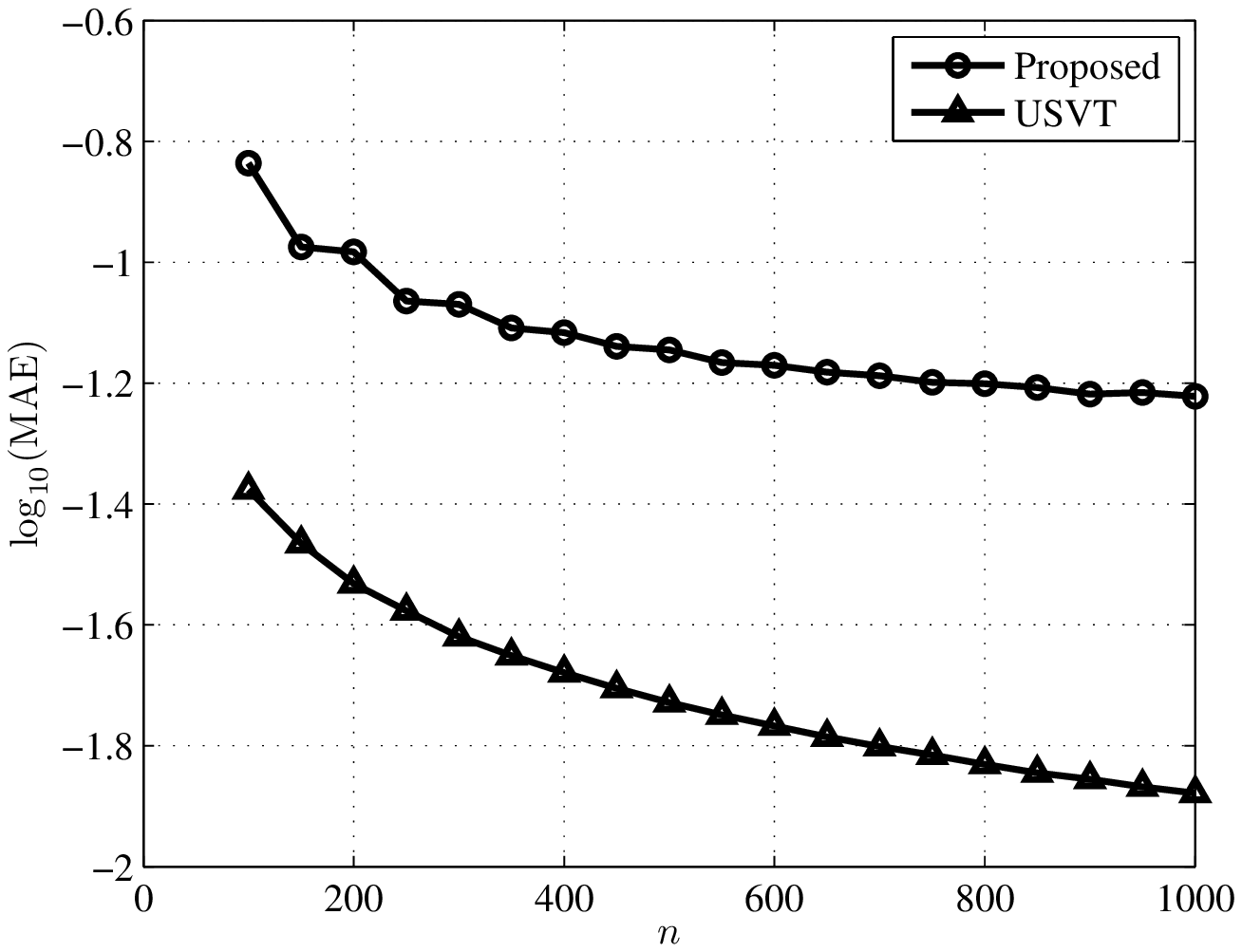}\\
(a) graphon $w_1$ & (b) graphon $w_2$
\end{tabular}
\caption{Comparison between SBA and USVT in estimating two continuous graphons $w_1$ and $w_2$. Evidently, SBA performs better for $w_1$ (high-rank) and worse for $w_2$ (low-rank).}
\label{fig:Plot5 and Plot6}
\end{figure}

\section{Concluding remarks}
\label{sec:conclusion}
We presented a new computational tool for estimating graphons. The proposed algorithm approximates the continuous graphon by a stochastic block-model, in which the first step is to cluster the unknown vertex labels into blocks by using an empirical estimate of the distance between two graphon slices, and the second step is to build an empirical histogram to estimate the graphon. Complete consistency analysis of the algorithm is derived. The algorithm was evaluated experimentally, and we found that the algorithm is effective in estimating block structured graphons.

\textbf{Code}. An implementation of the stochastic blockmodel approximation (SBA) algorithm proposed in this paper is available online at: \url{https://github.com/airoldilab/SBA}

\textbf{Acknowledgments}. EMA is partially supported by NSF CAREER award IIS-1149662, ARO MURI award W911NF-11-1-0036, and an Alfred P. Sloan Research Fellowship. SHC is partially supported by a Croucher Foundation Post-Doctoral Research Fellowship.

\footnotesize

\appendix
\section{Proofs for Section 3.1}
\setcounter{theorem}{0}
\begin{theorem}
\label{thm:unbiasedness of dhatij}
The estimator $\dhat_{ij}$ for $d_{ij}$ is unbiased. Further, for any $\epsilon > 0$, if the graph is directed, then
\begin{equation}
\Pr\left[ \left|\dhat_{ij}  - d_{ij}\right| > \epsilon \right] \le 8 e^{-\frac{S\epsilon^2}{32/T + 8\epsilon/3}},
\end{equation}
and if the graph is un-directed, then
\begin{equation}
\Pr\left[ \left|\dhat_{ij}  - d_{ij}\right| > \epsilon \right] \le 8 e^{-\frac{S\epsilon^2}{64/T + 8\epsilon/3}},
\end{equation}
where $S$ is the size of the sampling neighborhood $\calS$, and $2T$ is the number of observations.
\end{theorem}

\begin{proof}
First, for given $u_i$ and $u_j$, let us define the following two quantities
\begin{align*}
c_{ij} &\defequal \int_0^1 w(x,u_i)w(x,u_j) dx,\\
r_{ij} &\defequal \int_0^1 w(u_i,y)w(u_j,y) dy.
\end{align*}
Consequently, we express $d_{ij}$ as
\begin{align*}
d_{ij}
&\defequal \frac{1}{2} \left( \int_0^1 (w(u_i,y)-w(u_j,y))^2 dy + \int_0^1 (w(x,u_i)-w(x,u_j))^2 dx\right)\\
&= \frac{1}{2} \left[ \left(r_{ii} - r_{ij} - r_{ji} + r_{jj}\right) + \left(c_{ii} - c_{ij} - c_{ji} + c_{jj}\right) \right].
\end{align*}

In order to study $\dhat_{ij}$ (the estimator of $d_{ij}$), it is desired to express $\dhat_{ij}$ in the same form of $d_{ij}$:
\begin{align}
\dhat_{ij}
&= \frac{1}{S} \sum_{k \in \calS}\left\{
\frac{1}{2}\left[
\left(\rhat_{ii}^k - \rhat_{ij}^k - \rhat_{ji}^k + \rhat_{jj}^k \right) +
\left(\chat_{ii}^k - \chat_{ij}^k - \chat_{ji}^k + \chat_{jj}^k \right)\right]
\right\},
\label{eq:thm1,proof,dhat}
\end{align}
where $\calS = \{1,\ldots,n\} \backslash \{i,j\}$ is the sampling neighborhood, and $S = |\calS|$. In \eref{eq:thm1,proof,dhat}, individual components are defined as
\begin{align*}
\chat_{ij}^k &= \frac{1}{T^2}\left( \sum_{1\le t_1 \le T}G_{t_1}[k,i]\right)\left( \sum_{T < t_2 \le 2T}G_{t_2}[k,j]\right),\\
\rhat_{ij}^k &= \frac{1}{T^2}\left( \sum_{1\le t_1 \le T}G_{t_1}[i,k]\right)\left( \sum_{T < t_2 \le 2T}G_{t_2}[j,k]\right).
\end{align*}
Thus, if we can show that $\rhat^k_{ij}$ and $\chat^k_{ij}$ are unbiased estimators of $r_{ij}$ and $c_{ij}$, \emph{i.e.}, $\E[\rhat^k_{ij}] = r_{ij}$ and $\E[\chat^k_{ij}] = c_{ij}$, then by linearity of expectation, $\dhat_{ij}$ will be an unbiased estimator of $d_{ij}$.

To this end, we consider the conditional expectation of $G_{t_1}[i,k]G_{t_2}[j,k]$ given $u_k$:
\begin{align}
\E[ G_{t_1}[i,k]G_{t_2}[j,k] \;|\; u_k]
&= 1 \cdot \Pr\Big[G_{t_1}[i,k]G_{t_2}[j,k] = 1  \;\Big|\; u_k \Big] + 0 \cdot \Pr\Big[G_{t_1}[i,k]G_{t_2}[j,k] = 1  \;\Big|\; u_k \Big] \notag\\
&= \Pr\Big[G_{t_1}[i,k] = 1 \;\mathrm{and}\; G_{t_2}[j,k] = 1 \;\Big|\; u_k \Big] \notag \\
&= \Pr[G_{t_1}[i,k] = 1 \;|\; u_k] \cdot \Pr[G_{t_2}[j,k] = 1 \;|\; u_k], \quad \mathrm{because}\;\; G_{t_1}[i,k] \,\perp\, G_{t_2}[j,k] \notag\\
&= w(u_i,u_k)w(u_j,u_k). \label{eq:thm1,proof,E[Gt1 Gt2|u]}
\end{align}
Therefore,
\begin{align}
\E\Big[\rhat_{ij}^k \;|\; u_k\Big] &= \frac{1}{T^2} \left( \sum_{t_2 = T+1}^{2T} \sum_{t_1 = 1}^T \E\Big[G_{t_1}[i,k] G_{t_2}[j,k] \;|\; u_k \Big] \right) \notag\\
&= \frac{1}{T^2} \left( \sum_{t_2 = T+1}^{2T} \sum_{t_1 = 1}^T w(u_i,u_k)w(u_j,u_k) \right), \quad\mbox{by substituting \eref{eq:thm1,proof,E[Gt1 Gt2|u]}} \notag\\
&= w(u_i,u_k)w(u_j,u_k). \label{eq:thm1,proof,E[r|u]}
\end{align}
Then, by the law of iterated expectations, we have
\begin{align}
\E\left[\rhat_{ij}^k \right] &= \E \left[ \E \left[\rhat_{ij}^k \;|\; u_k \right] \right] \notag \\
&= \E\Big[ w(u_i,u_k)w(u_j,u_k) \Big], \quad\mbox{by substituting \eref{eq:thm1,proof,E[r|u]}} \notag\\
&= \int_0^1 w(u_i,v)w(u_j,v) dv, \quad\mathrm{because}\; u_k \sim \mbox{Uniform}(0,1) \notag\\
&= r_{ij}.
\end{align}
Therefore, $\rhat_{ij}^k$ is an unbiased estimator of $r_{ij}$. The proof of $\chat_{ij}$ can be similarly proved by switching roles of $G_t[i,k]$ to $G_t[k,i]$. Since $\rhat_{ij}^k$ and $\chat_{ij}^k$ are both unbiased, $\dhat_{ij}$ must be unbiased.

\vspace{2cm}

Now we proceed to prove the second part of the theorem. We first claim that
\begin{align}
\Var\left[\rhat_{ij}^k \right] \le 2/T \quad\mbox{and} \quad \Var\left[\chat_{ij}^k \right] \le 2/T.
\end{align}
To prove this, we note that
\begin{align*}
\Var\left[\rhat_{ij}^k \right] &= \Var\left[ \sum_{t_2=T+1}^{2T}  \sum_{t_1=1}^T G_{t_1}[ik] G_{t_2}[jk] \right]\\
&= \sum_{t_2=T+1}^{2T} \sum_{t_1=1}^T \Var\Big[  G_{t_1}[ik] G_{t_2}[jk] \Big] \\
&\quad + \sum_{\substack{\tau_2=T+1 \\ \tau_2\not=t_2 }}^{2T} \sum_{t_2=T+1}^{2T} \sum_{ \substack{\tau_1=1\\ \tau_1\not=t_1}}^{T}  \sum_{t_1=1}^T \Cov\Big[  G_{t_1}[ik] G_{t_2}[jk],\; G_{\tau_1}[ik] G_{\tau_2}[jk] \Big]
\end{align*}

\noindent We consider three cases:

\noindent \textbf{Case 1}. First assume $\tau_1 \not= t_1$ and $\tau_2 \not= t_2$. (Occurs $(T-1)^2T^2$ times.)
\begin{align}
&\Cov\Big[  G_{t_1}[ik] G_{t_2}[jk],\; G_{\tau_1}[ik] G_{\tau_2}[jk] \Big] \notag \\
&= \E\Big[ \left(G_{t_1}[ik] G_{t_2}[jk] - \E[G_{t_1}[ik] G_{t_2}[jk]]\right) \left(G_{\tau_1}[ik] G_{\tau_2}[jk] - \E[G_{\tau_1}[ik] G_{\tau_2}[jk]] \right)    \Big] \notag\\
&= \E\Big[ \left(G_{t_1}[ik] G_{t_2}[jk] - w_{ik}w_{jk}\right) \left(G_{\tau_1}[ik] G_{\tau_2}[jk] - w_{ik}w_{jk}\right) \Big] \notag\\
&= \E\Big[ G_{t_1}[ik] G_{t_2}[jk]G_{\tau_1}[ik] G_{\tau_2}[jk] \Big] - \E\Big[ G_{\tau_1}[ik] G_{\tau_2}[jk] \Big]w_{ik}w_{jk} - \E\Big[ G_{t_1}[ik] G_{t_2}[jk] \Big]w_{ik}w_{jk}  + w_{ik}^2w_{jk}^2\notag\\
&= \E\Big[ G_{t_1}[ik] G_{t_2}[jk]G_{\tau_1}[ik] G_{\tau_2}[jk] \Big] - w_{ik}^2w_{jk}^2
\label{eq:Cov}
\end{align}
The first term in \eref{eq:Cov} is $\E\Big[ G_{t_1}[ik] G_{t_2}[jk]G_{\tau_1}[ik] G_{\tau_2}[jk] \Big] = w_{ik}^2w_{jk}^2$ because $G_{t_1}[ik]$, $G_{t_2}[jk]$, $G_{\tau_1}[ik]$ and $G_{\tau_2}[jk]$ are all independent. Therefore, the overall sum in \eref{eq:Cov} is 0.

\noindent \textbf{Case 2}. Next assume that $\tau_1 \not= t_1$ but $\tau_2 = t_2$. (Occurs $(T-1)T^2$ times.) In this case,
\begin{align*}
\E\Big[ G_{t_1}[ik] G_{t_2}[jk]G_{\tau_1}[ik] G_{\tau_2}[jk] \Big]
&= \E\Big[ G_{t_1}[ik] \Big] \E\Big[G_{\tau_1}[ik]\Big] \E\Big[G_{t_2}[jk] G_{\tau_2}[jk] \Big] \\
&= w_{ik} w_{ik} \E\Big[G_{t_2}[jk]^2 \Big]\\
&= w_{ik}^2 w_{jk}.
\end{align*}
Substituting this result into \eref{eq:Cov} yields the covariance
\begin{align*}
\Cov\Big[  G_{t_1}[ik] G_{t_2}[jk],\; G_{\tau_1}[ik] G_{\tau_2}[jk] \Big]
&= w_{ik}^2 w_{jk} - w_{ik}^2w_{jk}^2 = w_{ik}^2 w_{jk}(1 - w_{jk}) \le 1.
\end{align*}

\noindent \textbf{Case 3}. Assume $\tau_1 = t_1$ but $\tau_2 \not= t_2$. (Occurs $(T-1)T^2$ times.) In this case,
\begin{align*}
\E\Big[ G_{t_1}[ik] G_{t_2}[jk]G_{\tau_1}[ik] G_{\tau_2}[jk] \Big] = w_{ik} w_{jk}^2,
\end{align*}
and so the covariance becomes
\begin{align*}
\Cov\Big[  G_{t_1}[ik] G_{t_2}[jk],\; G_{\tau_1}[ik] G_{\tau_2}[jk] \Big] = w_{ik} w_{jk}^2(1 - w_{ik}) \le 1.
\end{align*}

\noindent Combining all 3 cases, we have the following bound:
\begin{align*}
\Var[\widehat{r}_{ij}^k ] &= \frac{1}{T^4} \Var\left[ \sum_{t_1} \sum_{t_2} G_{t_1}[ik] G_{t_2}[jk] \right] \\
&= \frac{1}{T^4} \left[ \sum_{t_1}\sum_{t_2} \Var\Big[  G_{t_1}[ik] G_{t_2}[jk] \Big] + (T-1)T^2 w_{ik}^2 w_{jk} (1 - w_{jk}) + (T-1)T^2 w_{ik} w_{jk}^2 (1 - w_{ik})\right]\\
&= \frac{1}{T^4} \left[ T^2 w_{ik}w_{jk}(1 - w_{ik}w_{jk}) + (T-1)T^2 w_{ik}^2 w_{jk} (1 - w_{jk}) + (T-1)T^2 w_{ik} w_{jk}^2 (1 - w_{ik}) \right] \\
&\le \frac{1}{T^4} \left[ T^2 + 2(T-1)T^2 \right] \\
&= \frac{2T-1}{T^2} \le \frac{2}{T}.
\end{align*}
The bound for $\Var\left[\chat_{ij}^k \right]$ can be proved similarly.

Next, we observe that $G_{t}$ (for any $t$) is a directed graph. So the random variables $G_{t_1}[i,k]$ and $G_{t_1}[k,i]$ are independent. Similarly, $G_{t_2}[j,k]$ and $G_{t_2}[k,j]$ are independent. Therefore, the product variables $G_{t_1}[i,k]G_{t_2}[j,k]$ and $G_{t_1}[k,i]G_{t_2}[k,j]$ must be independent for any fixed $u_i$, $u_j$ and $u_k$, where $i\not=j$ and $k = \{1,\ldots,n\} \backslash \{i,j\}$. Consequently, $\rhat_{ij}^k$ and $\chat_{ij}^k$ are independent, and hence
\begin{align*}
\E[\rhat_{ij}^k\chat_{ij}^k]
&= \E\left[\rhat_{ij}^k \right] \cdot \E\left[ \chat_{ij}^k \right]\\
&= r_{ij} c_{ij},
\end{align*}
which implies that $\rhat_{ij}^k$ and $\chat_{ij}^k$ are uncorrelated: $\E\left[ (\rhat_{ij}^k - r_{ij}) (\chat_{ij}^k-c_{ij})\right] = 0$. Consequently,
\begin{align*}
\Var\left[ \frac{1}{2} \left(\rhat_{ij}^k + \chat_{ij}^k\right)\right] = \frac{1}{4} \left(\Var\left[\rhat_{ij}^k\right] + \Var\left[\chat_{ij}^k\right]\right) \le \frac{1}{T}.
\end{align*}

Since $\rhat_{ij}  =  \frac{1}{S} \sum\limits_{k \in \calS} \rhat_{ij}^k$ and $\chat_{ij}  =  \frac{1}{S} \sum\limits_{k \in \calS} \chat_{ij}^k$, by Bernstein's inequality we have
\begin{align*}
\Pr\left[ \left| \frac{1}{2} \left(\rhat_{ij}+\chat_{ij}\right)  - \frac{1}{2}\left(r_{ij}+c_{ij}\right) \right| > \epsilon \right]
&= \Pr\left[ \left| \frac{1}{S} \sum\limits_{k \in \calS} \frac{1}{2} \left( \rhat_{ij}^k + \chat_{ij}^k\right)  - \frac{1}{2}\left(r_{ij}+c_{ij}\right) \right| > \epsilon \right]  \\
&\le 2e^{-\frac{S\epsilon^2}{2 \left( \Var\left[\frac{1}{2} \left(\rhat_{ij}^k + \chat_{ij}^k\right)\right] + \epsilon/3 \right)}}
\le 2e^{-\frac{S\epsilon^2}{2( 1/T + \epsilon/3)}}.
\end{align*}
Finally, we note that
\begin{align*}
|\dhat_{ij} - d_{ij}|
&\le   \frac{1}{2}\left|\rhat_{ii}+\chat_{ii}-r_{ii}-c_{ii}\right| + \frac{1}{2}\left|\rhat_{ij}+\chat_{ij}-r_{ij}-c_{ij}\right| + \\
&\quad \frac{1}{2}\left|\rhat_{ji}+\chat_{ji}-r_{ji}-c_{ji}\right| + \frac{1}{2}\left|\rhat_{jj}+\chat_{jj}-r_{jj}-c_{jj}\right|.
\end{align*}
Therefore by union bound we have
\begin{align*}
&\Pr[|\dhat_{ij} - d_{ij}|  > \epsilon ] \\
&\le \Pr\Big[ \frac{1}{2}\left|\rhat_{ii}+\chat_{ii}-r_{ii}-c_{ii}\right| + \frac{1}{2}\left|\rhat_{ij}+\chat_{ij}-r_{ij}-c_{ij}\right| + \\
&\quad\quad\quad+ \frac{1}{2}\left|\rhat_{ji}+\chat_{ji}-r_{ji}-c_{ji}\right| + \frac{1}{2}\left|\rhat_{jj}+\chat_{jj}-r_{jj}-c_{jj}\right|
 > \epsilon \Big] \\
&\le \Pr\Big[ \left|\frac{1}{2}\left(\rhat_{ii}+\chat_{ii}\right)-\frac{1}{2}\left(r_{ii}+c_{ii}\right)\right| > \epsilon/4 \Big] +  \Pr\Big[ \left|\frac{1}{2}\left(\rhat_{ij}+\chat_{ij}\right)-\frac{1}{2}\left(r_{ij}+c_{ij}\right)\right| > \epsilon/4 \Big] + \\
&\quad+\Pr\Big[ \left|\frac{1}{2}\left(\rhat_{ji}+\chat_{ji}\right)-\frac{1}{2}\left(r_{ji}+c_{ji}\right)\right| > \epsilon/4 \Big] +
\Pr\Big[ \left|\frac{1}{2}\left(\rhat_{jj}+\chat_{jj}\right)-\frac{1}{2}\left(r_{jj}+c_{jj}\right)\right| > \epsilon/4 \Big]\\
&\le 8 e^{-\frac{S\epsilon^2/16}{2(1/T + \epsilon/12)}} = 8 e^{-\frac{S\epsilon^2}{32/T + 8\epsilon/3}}.
\end{align*}

If the graph is un-directed, then $c_{ij}^k = r_{ij}^k$ and we can only have $\Var\left[ \frac{1}{2}\left(r_{ij}^k+c_{ij}^k \right)\right] \le \frac{2}{T}$ instead of $\Var\left[ \frac{1}{2}\left(r_{ij}^k+c_{ij}^k \right)\right] \le \frac{1}{T}$. In this case,
\begin{align*}
&\Pr[|\dhat_{ij} - d_{ij}|  > \epsilon ] \le 8 e^{-\frac{S\epsilon^2}{64/T + 8\epsilon/3}}.
\end{align*}
\end{proof}

\section{Proofs for Section 3.2}
\begin{theorem}
\label{thm:number of blocks}
Let $\Delta$ be the accuracy parameter and $K$ be the number of blocks estimated by Algorithm 1, then
\begin{equation}
\Pr\left[ K > \frac{QL\sqrt{2}}{\Delta} \right] \le 8n^2 e^{-\frac{S\Delta^4}{128/T + 16\Delta^2/3}},
\end{equation}
where $L$ is the Lipschitz constant and $Q$ is the number of Lipschitz blocks in the ground truth $w$.
\end{theorem}

\begin{proof}
Recall that in defining the Lipschitz condition of $w$ (Section 2.1), we defined a sequence of non-overlapping intervals $I_k = [\alpha_k,\,\alpha_{k+1}]$, where $0 = \alpha_0 < \ldots < \alpha_Q = 1$, and $Q$ is the number of Lipschitz blocks of $w$. For each of the interval $I_k$, we divide it into $R \defequal \frac{L\sqrt{2}}{\Delta}$ subintervals of equal size $1/R$. Thus, the distance between any two elements in the same subinterval is at most $1/R$. Also, the total number of subintervals over $[0,1]$ is $QR$.

Now, suppose that there are $K > QR = \frac{QL\sqrt{2}}{\Delta}$ blocks defined by the algorithm, and denote the $K$ pivots be $p_1,\ldots,p_K$. By the pigeonhole principle, there must be at least two pivots $p_i$ and $p_j$ in the same sub-interval. In this case, the distance $d_{p_i,p_j}$ must satisfy the following condition:
\begin{align*}
d_{p_i,p_j}
&= \frac{1}{2} \left(\int_0^1 (w(x,u_{p_i})-w(x,u_{p_j}))^2 dx  + \int_0^1 (w(u_{p_i},y)-w(u_{p_j},y))^2 dy\right)\\
&\le L^2 (u_{p_i} - u_{p_j})^2 \\
&\le L^2 \frac{1}{R^2} =  \frac{\Delta^2}{2}.
\end{align*}
However, from the algorithm it holds that $\dhat_{p_i,p_j} \ge \Delta^2$. So, if $K > QR$, then $\dhat_{p_i,p_j} - d_{p_i,p_j} > \frac{\Delta^2}{2}$.

Let $\calE$ be the following event:
\begin{equation*}
\calE = \left\{ \dhat_{p_i,p_j} - d_{p_i,p_j} > \frac{\Delta^2}{2} \quad\mbox{for at least one pair of } p_i,p_j \right\}.
\end{equation*}
Then, since the event $\calE$ is a consequence of the event $\{K > QR\}$, we have
\begin{equation*}
\Pr\left[ K > \frac{QL\sqrt{2}}{\Delta} \right] = \Pr[ K > QR ] \le \Pr[\;\calE\;].
\end{equation*}
To bound $\Pr[\calE]$, we observe that
\begin{align*}
\Pr\left[ \dhat_{p_i,p_j} - d_{p_i,p_j} > \frac{\Delta^2}{2} \;\Big|\; p_i, p_j \right]
\le 8 e^{-\frac{S(\Delta^2/2)^2}{32/T + 8(\Delta^2/2)/3}} = 8 e^{-\frac{S\Delta^4}{128/T + 16\Delta^2/3}}.
\end{align*}
Therefore, by union bound, 
\begin{align*}
\Pr\left[ \calE \;\Big|\; p_1,\ldots,p_K\right]
&\le \sum_{p_i,p_j} \Pr\left[ \dhat_{p_i,p_j} - d_{p_i,p_j} > \frac{\Delta^2}{2} \;\Big|\; p_i, p_j  \right]\\
&\le 8n^2 e^{-\frac{S\Delta^4}{128/T + 16\Delta^2/3}},
\end{align*}
and hence,
\begin{align*}
\Pr\left[ \; \calE \;\right]
&= \sum_{p_1,\ldots,p_K} \Pr\left[ \calE \,|\, p_1,\ldots,p_K \right] \Pr\left[ p_1,\ldots,p_K \right]\\
&\le \left(8n^2 e^{-\frac{S\Delta^4}{128/T + 16\Delta^2/3}}\right) \cdot \sum_{p_1,\ldots,p_K} \Pr\left[ p_1,\ldots,p_K \right]\\
&= 8n^2 e^{-\frac{S\Delta^4}{128/T + 16\Delta^2/3}}.
\end{align*}
This completes the proof.
\end{proof}

\section{Proofs for Section 3.3}
\setcounter{lemma}{0}
\begin{lemma}
\label{lemma:bound wbar and w}
Let $\Bhat_i = \{i_1,i_2,\ldots,i_{|\Bhat_i|}\}$ and $\Bhat_j = \{j_1,j_2,\ldots,j_{|\Bhat_j|}\}$ be two clusters returned by the Algorithm. Suppose that $\{u_{i_1},u_{i_2},\ldots,u_{i_{|\Bhat_i|}}\}$ and $\{u_{j_1},u_{j_2},\ldots,u_{j_{|\Bhat_j|}}\}$ are the ground truth labels of the vertices in $\Bhat_i$ and $\Bhat_j$, respectively. Let
\begin{equation}
\wbar_{ij} = \frac{1}{|\Bhat_i||\Bhat_j|} \sum_{i_x \in \Bhat_i} \sum_{j_x \in \Bhat_j} w(u_{i_x}, u_{j_x}).
\end{equation}
Assume that the precision parameter satisfies $\Delta^2 < \frac{\delta^2 L}{4}$, where $L$ is the Lipschitz constant and $\delta$ is the size of the smallest Lipschitz interval. Then, for any $i_v \in \Bhat_i$ and $j_v \in \Bhat_j$,
\begin{equation}
\Pr\left[ |\wbar_{ij} - w(u_{i_v,j_v})| > 8\Delta^{1/2} L^{1/4} \right] \le 32 |\Bhat_i| |\Bhat_j| e^{-\frac{S\Delta^4}{32/T + 8\Delta^2/3}}.
\end{equation}
\end{lemma}

\begin{proof}
Let $i_p \in \Bhat_i$ and $j_p \in \Bhat_j$ be pivots of the clusters $\Bhat_i$ and $\Bhat_j$, respectively. By definition of pivots, it holds that $|\dhat_{i_p,i_v}| \le \Delta^2$ and $|\dhat_{j_p,j_v}| \le \Delta^2$ for any vertices $i_v \in \Bhat_i$ and $j_v \in \Bhat_j$. Therefore,
\begin{align*}
\begin{array}{rrl}
                & 0  &\le -|\dhat_{i_p,i_v}| + \Delta^2 \le -\dhat_{i_p,i_v} + \Delta^2 \\
\Rightarrow     &    d_{i_p,i_v} &\le d_{i_p,i_v} - \dhat_{i_p,i_v} + \Delta^2 \le |d_{i_p,i_v} - \dhat_{i_p,i_v}| + \Delta^2,
\end{array}
\end{align*}
which implies that
\begin{align*}
\Pr\left[ d_{i_p,i_v} > 2\Delta^2 \right]
&\le \Pr\left[ |d_{i_p,i_v}-\dhat_{i_p,i_v}| + \Delta^2 > 2\Delta^2 \right] \\
&=   \Pr\left[ |d_{i_p,i_v}-\dhat_{i_p,i_v}| > \Delta^2\right]\\
&\le 8e^{-\frac{S\Delta^4}{32/T + 8\Delta^2/3}}.
\end{align*}
Similarly, we have $\Pr\left[ d_{j_p,j_v} > 2\Delta^2 \right] \le 8e^{-\frac{S\Delta^4}{32/T + 8\Delta^2/3}}$. Thus,
\begin{align*}
\Pr\left[d_{i_p,i_v} > 2\Delta^2 \;\cup\; d_{j_p,j_v} > 2\Delta^2 \right]
&\le \Pr\left[d_{i_p,i_v} > 2\Delta^2\right] + \Pr\left[d_{j_p,j_v} > 2\Delta^2 \right]\\
&\le 16 e^{-\frac{S\Delta^4}{32/T + 8\Delta^2/3}}.
\end{align*}

Let $d_{ij}^c = \int_0^1 (w(x,u_i)-w(x,u_j))^2 dx$ and $d_{ij}^r = \int_0^1 (w(u_i,y)-w(u_j,y))^2 dy$. By Lemma \ref{lemma:consistency1}, it holds that for any $0 < \left(\epsilon/2\right)^2 < 2\delta L$, if $d_{i,j}^c \le \frac{ (\epsilon/2)^4}{8L} = \frac{\epsilon^4}{128L}$ and $d_{i,j}^r \le \frac{\epsilon^4}{128L}$, then
\begin{align*}
\sup_{x \in [0,1]} |w(x,u_i) - w(x,u_j)| \le \frac{\epsilon}{2},\\
\sup_{y \in [0,1]} |w(u_i,y) - w(u_j,y)| \le \frac{\epsilon}{2}.
\end{align*}
Therefore, if $d_{i_p,i_v}^c \le \frac{\epsilon^4}{128L}$, $d_{i_p,i_v}^r \le \frac{\epsilon^4}{128L}$, $d_{j_p,j_v}^c \le \frac{\epsilon^4}{128L}$ and $d_{j_p,j_v}^r \le \frac{\epsilon^4}{128L}$, then for pivots $i_p \in \Bhat_{i}$, $j_p \in \Bhat_j$, and vertex $i_v \in \Bhat_i$, $j_v \in \Bhat_j$:
\begin{align}
|w(u_{i_v}, u_{j_v}) - w(u_{i_p},u_{j_p})|
&\le |w(u_{i_v}, u_{j_v}) - w(u_{i_v},u_{j_p})| + |w(u_{i_v},u_{j_p}) - w(u_{i_p}, u_{j_p})| \notag \\
&\le \sup_{x \in [0,1]} |w(x,u_{j_v}) - w(x,u_{j_p})| + \sup_{y \in [0,1]} |w(u_{i_v},y) - w(u_{j_p},y)| \notag \\
&= \frac{\epsilon}{2} + \frac{\epsilon}{2} = \epsilon. \label{eq:lemma,bound wbar and w,eq1}
\end{align}
Also, if $d_{i_p,i_x}^c \le \frac{\epsilon^4}{128L}$, $d_{i_p,i_x}^r \le \frac{\epsilon^4}{128L}$, $d_{j_p,j_x}^c \le \frac{\epsilon^4}{128L}$ and $d_{j_p,j_x}^r \le \frac{\epsilon^4}{128L}$ for vertex every $i_x \in \Bhat_i$, $j_x \in \Bhat_j$
\begin{align}
&\left| \frac{1}{|\Bhat_i||\Bhat_j|} \sum_{i_x \in \Bhat_i} \sum_{j_x \in \Bhat_j} w(u_{i_x}, u_{j_x}) - w(u_{i_p},u_{j_p})\right| \notag\\
&\le \left| \frac{1}{|\Bhat_i||\Bhat_j|} \sum_{i_x \in \Bhat_i} \sum_{j_x \in \Bhat_j} w(u_{i_x}, u_{j_x}) -\frac{1}{|\Bhat_i|} \sum_{i_x \in \Bhat_i} w(u_{i_x},u_{j_p}) \right|
+ \left| \frac{1}{|\Bhat_i|} \sum_{i_x \in \Bhat_i} w(u_{i_x},u_{j_p}) - w(u_{i_p},u_{j_p})\right| \notag \\
&\le \frac{1}{|\Bhat_i|}\frac{1}{|\Bhat_j|}  \sum_{i_x \in \Bhat_i} \sum_{j_x \in \Bhat_j} \left| w(u_{i_x}, u_{j_x}) - w(u_{i_x},u_{j_p}) \right|
+ \frac{1}{|\Bhat_i|}  \sum_{i_x \in \Bhat_i}  \left| w(u_{i_x},u_{j_p}) - w(u_{i_p},u_{j_p}) \right| \notag \\
&\le \frac{1}{|\Bhat_i|}\frac{1}{|\Bhat_j|} \sum_{i_x \in \Bhat_i} \sum_{j_x \in \Bhat_j} \frac{\epsilon}{2} +  \frac{1}{|\Bhat_i|}  \sum_{i_x \in \Bhat_i} \frac{\epsilon}{2} = \epsilon.
\label{eq:lemma,bound wbar and w,eq2}
\end{align}
Combining \eref{eq:lemma,bound wbar and w,eq1} and \eref{eq:lemma,bound wbar and w,eq2} with triangle inequality yields
\begin{align*}
\left| \frac{1}{|\Bhat_i||\Bhat_j|} \sum_{i_x \in \Bhat_i} \sum_{j_x \in \Bhat_j} w(u_{i_x}, u_{j_x}) - w(u_{i_v},u_{j_v})\right| \le 2\epsilon.
\end{align*}
Consequently, by contrapositive this implies that
\begin{align*}
&\quad           \left| \wbar_{ij} - w(u_{i_v},u_{j_v})\right| > 2\epsilon \\
&\Rightarrow \bigcup_{i_x \in \Bhat_i, j_x \in \Bhat_j} \left( \; d_{i_p,i_x}^c > \frac{\epsilon^4}{128L} \;\cup\; d_{i_p,i_x}^r > \frac{\epsilon^4}{128L} \;\cup\; d_{j_p,j_x}^c > \frac{\epsilon^4}{128L} \;\cup\; d_{j_p,j_x}^r > \frac{\epsilon^4}{128L}\right)\\
&\Rightarrow  \bigcup_{i_x \in \Bhat_i, j_x \in \Bhat_j} \left(  \; d_{i_p,i_x} > \frac{\epsilon^4}{128L} \;\cup\; d_{j_p,j_x} > \frac{\epsilon^4}{128L}\right).
\end{align*}
Therefore,
\begin{align*}
\Pr\left[ \left| \wbar_{ij} - w(u_{i_v},u_{j_v})\right| > 2\epsilon \right]
&\le  \Pr\left[\bigcup_{i_x \in \Bhat_i, j_x \in \Bhat_j} \left( d_{i_p,i_x} > \frac{\epsilon^4}{128L} \;\cup\; d_{j_p,j_x} > \frac{\epsilon^4}{128L} \right)\right] \\
&\le  \sum_{i_x \in \Bhat_i, j_x \in \Bhat_j}\left(\Pr\left[d_{i_p,i_x} > \frac{\epsilon^4}{128L}\right] + \Pr\left[d_{j_p,j_x} > \frac{\epsilon^4}{128L} \right]\right).
\end{align*}

Assuming $\Delta < \delta\sqrt{L}/2$ and setting $\epsilon = 4\Delta^{1/2}L^{1/4}$, we have $0 < \left(\epsilon/2\right)^2 < 2\delta L$ and thus
\begin{align*}
\Pr\left[ \left| \wbar_{ij} - w(u_{i_v},u_{j_v})\right| > 8\Delta^{1/2} L^{1/4} \right]
&\le  \sum_{i_x \in \Bhat_i, j_x \in \Bhat_j}\left(  \Pr\left[d_{i_p,i_x} > 2\Delta^2 \right] + \Pr\left[d_{j_p,j_x} > 2\Delta^2 \right]\right)\\
&\le  32  |\Bhat_i| |\Bhat_j| e^{-\frac{S\Delta^4}{32/T + 8\Delta^2/3}}.
\end{align*}
\end{proof}

\begin{lemma}
\label{lemma:bound what and wbar}
Let $\Bhat_i = \{i_1,i_2,\ldots,i_{|\Bhat_i|}\}$ and $\Bhat_j = \{j_1,j_2,\ldots,j_{|\Bhat_j|}\}$ be two clusters returned by the Algorithm. Suppose that $\{u_{i_1},u_{i_2},\ldots,u_{i_{|\Bhat_i|}}\}$ and $\{u_{j_1},u_{j_2},\ldots,u_{j_{|\Bhat_j|}}\}$ are the ground truth labels of the vertices in $\Bhat_i$ and $\Bhat_j$, respectively. Let
\begin{align*}
\what_{ij} &= \frac{1}{|\Bhat_i||\Bhat_j|} \sum_{i_x \in \Bhat_i} \sum_{j_x \in \Bhat_j} \left( \frac{G_1[i_x,j_x] + \ldots + G_{2T}[i_x,j_x]}{2T}\right),\\
\wbar_{ij} &= \frac{1}{|\Bhat_i||\Bhat_j|} \sum_{i_x \in \Bhat_i} \sum_{j_x \in \Bhat_j} w(u_{i_x}, u_{j_x}).
\end{align*}
Then,
\begin{align*}
\Pr\left[ |\what_{ij} - \wbar_{ij}| >  8\Delta^{1/2} L^{1/4}  \right] \le 2 e^{-256 (T |\Bhat_i|\,|\Bhat_j| \sqrt{L} \Delta)} + 32 |\Bhat_i|^2 |\Bhat_j|^2 e^{-\frac{S\Delta^4}{32/T + 8\Delta^2/3}}.
\end{align*}
\end{lemma}

\begin{proof}
There are two possible situations that we need to consider. 

\noindent \textbf{Case 1}: For any vertex $i_v \in \Bhat_i$ and $j_v \in \Bhat_j$, the estimate of the previous lemma $\wbar_{ij}$ (independent of $(i_v,j_v)$) is close to the ground truth $w_{ij} \defequal w(u_{i_v},u_{j_v})$. In other words, we want $w(u_{i_v},u_{j_v})$ to stay close for all $i_v \in \Bhat_i$ and $j_v \in \Bhat_j$, so that the difference $|w_{ij} - \wbar_{ij}|$ remains small for all $i_v \in \Bhat_i$ and $j_v \in \Bhat_j$.

\noindent \textbf{Case 2}: Complement of case 1.

To encapsulate these two cases, we first define the event
\begin{equation*}
\calE = \left\{ |w_{ij} - \wbar_{ij}| \le  8\Delta^{1/2} L^{1/4} , \forall i_v \in \Bhat_i, \; j_v \in \Bhat_j\right\}
\end{equation*}
and define $\overline{\calE}$ be the complement of $\calE$. Then,
\begin{align*}
\Pr\left[ |\what_{ij} - \wbar_{ij}| > 8\Delta^{1/2} L^{1/4} \right]
&= \Pr\left[  |\what_{ij} - \wbar_{ij}| > 8\Delta^{1/2} L^{1/4} \;\Big|\; \calE \right] \Pr\left[\calE\right]\\
&\quad + \Pr\left[ | \what_{ij} - \wbar_{ij}| > 8\Delta^{1/2} L^{1/4} \;\Big|\; \overline{\calE} \right] \Pr\left[\overline{\calE}\right]\\
&\le \Pr\left[  |\what_{ij} - \wbar_{ij}| > 8\Delta^{1/2} L^{1/4} \;\Big|\; \calE \right] + \Pr\left[\overline{\calE}\right].
\end{align*}
So it remains to bound the two probabilities.

Conditioning on $\calE$, it holds that
\begin{equation*}
\wbar_{ij} - \epsilon \le w_{ij} \le \wbar + \epsilon.
\end{equation*}
Fix a vertex pair $(i_v,j_v)$, we note that $G_1[i_v,j_v],\ldots,G_{2T}[i_v,j_v]$ are independent Bernoulli random variable with common mean $w(u_{i_v},u_{j_v})$. Denote
\begin{equation*}
\what_{ij} = \frac{1}{2T|\Bhat_i||\Bhat_j|} \sum_{t = 1}^{2T }\sum_{i_x \in \Bhat_i} \sum_{j_x\in \Bhat_j} G_{t}[i_x,j_x],
\end{equation*}
then by Hoeffding inequality we have
\begin{align*}
\Pr\left[ \what_{ij} - \wbar_{ij} > 2\epsilon \;\Big|\; \calE \right]
&=   \Pr\left[ \what_{ij} > \wbar_{ij} + 2\epsilon \;\Big|\; \calE \right]\\
&\le \Pr\left[ \what_{ij} > w_{ij} + \epsilon \;\Big|\; \calE \right]\\
&\le e^{-2 (2T |\Bhat_i||\Bhat_j| \epsilon^2)},
\end{align*}
and similarly $\Pr\left[ \what_{ij} - \wbar_{ij} < -2\epsilon \;\Big|\; \calE \right] \le e^{-2 (2T |\Bhat_i|\,|\Bhat_j| \epsilon^2)}$. Therefore,
\begin{align*}
\Pr\left[ |\what_{ij} - \wbar_{ij}| > 2\epsilon \;\Big|\; \calE \right]
&\le 2 e^{-2 (2T |\Bhat_i|\,|\Bhat_j| \epsilon^2)}.
\end{align*}
Substituting $\epsilon = 4\Delta^{1/2} L^{1/4}$, we have
\begin{align*}
\Pr\left[ |\what_{ij} - \wbar_{ij}| > 8\Delta^{1/2} L^{1/4} \;\Big|\; \calE \right]
&\le 2 e^{-128 ( |\Bhat_i|\,|\Bhat_j| (2T) \sqrt{L} \Delta)}.
\end{align*}

The second probability is bounded as follows. Since $\overline{\calE}$ is the complement of $\calE$, it is bounded by the probability where at least one $(i_v,j_v)$ violates the condition. Therefore,
\begin{align*}
\Pr\left[\overline{\calE}\right]
&= \Pr\Big[ \mbox{at least one}\; i_v, j_v \;\mbox{s.t.}\;
\left| w(u_{i_v},u_{j_v}) - \wbar_{ij} \right| > 8\Delta^{1/2} L^{1/4} \Big] \\
&\le \sum_{i_v \in \Bhat_i} \sum_{j_v \in \Bhat_j} \Pr\left[ |w(u_{i_v},u_{j_v}) - \wbar_{ij}| > 8\Delta^{1/2} L^{1/4}  \right]\\
&\le 32 |\Bhat_i|^2 |\Bhat_j|^2 e^{-\frac{S\Delta^4}{32/T + 8\Delta^2/3}}.
\end{align*}

Finally, by combining the above results we have
\begin{align*}
\Pr\left[ |\what_{ij} - \wbar_{ij}| > 8\Delta^{1/2} L^{1/4} \right]  \le 2 e^{-256 (T |\Bhat_i|\,|\Bhat_j| \sqrt{L} \Delta)} + 32 |\Bhat_i|^2 |\Bhat_j|^2 e^{-\frac{S\Delta^4}{32/T + 8\Delta^2/3}}.
\end{align*}
\end{proof}

\begin{lemma}
\label{lemma:bound w and what}
Let $\Bhat_i = \{i_1,i_2,\ldots,i_{|\Bhat_i|}\}$ and $\Bhat_j = \{j_1,j_2,\ldots,j_{|\Bhat_j|}\}$ be two clusters returned by the Algorithm. Suppose that $\{u_{i_1},u_{i_2},\ldots,u_{i_{|\Bhat_i|}}\}$ and $\{u_{j_1},u_{j_2},\ldots,u_{j_{|\Bhat_j|}}\}$ are the ground truth labels of the vertices in $\Bhat_i$ and $\Bhat_j$, respectively. Let
\begin{align*}
\what_{ij} &= \frac{1}{|\Bhat_i||\Bhat_j|} \sum_{i_x \in \Bhat_i} \sum_{j_x \in \Bhat_j} \left( \frac{G_1[i_x,j_x] + \ldots + G_{2T}[i_x,j_x]}{2T}\right).
\end{align*}
Then,
\begin{align*}
\Pr\left[ |\what_{ij} - w_{ij}| > 16\Delta^{1/2} L^{1/4} \right]
\le 2 e^{-256 (T |\Bhat_i|\,|\Bhat_j| \sqrt{L} \Delta)} + 64 n^4 e^{-\frac{S\Delta^4}{32/T + 8\Delta^2/3}}.
\end{align*}
\end{lemma}

\begin{proof}
By Lemma \ref{lemma:bound wbar and w} and Lemma \ref{lemma:bound what and wbar}, we have
\begin{align*}
\Pr\left[ |\what_{ij} - \wbar_{ij}| > 8\Delta^{1/2} L^{1/4} \right]  &\le 2 e^{-256 (T |\Bhat_i|\,|\Bhat_j| \sqrt{L} \Delta)} + 32 |\Bhat_i|^2 |\Bhat_j|^2 e^{-\frac{S\Delta^4}{32/T + 8\Delta^2/3}}\\
\Pr\left[ \wbar_{ij} - w_{ij} > 8\Delta^{1/2} L^{1/4} \right]         &\le 32|\Bhat_i||\Bhat_j| e^{-\frac{S\Delta^4}{32/T + 8\Delta^2/3}}.
\end{align*}
Therefore, it follows that
\begin{align*}
\Pr\left[ |\what_{ij} - w_{ij}| > 16\Delta^{1/2} L^{1/4} \right]
&\le \Pr\left[ |\what_{ij} - \wbar_{ij}| > 8\Delta^{1/2} L^{1/4} \right] + \Pr\left[ \wbar_{ij} - w_{ij} > 8\Delta^{1/2} L^{1/4} \right]\\
&\le 2 e^{-256 (T |\Bhat_i|\,|\Bhat_j| \sqrt{L} \Delta)} + 32 |\Bhat_i|^2 |\Bhat_j|^2 e^{-\frac{S\Delta^4}{32/T + 8\Delta^2/3}} + 32|\Bhat_i| |\Bhat_j|  e^{-\frac{S\Delta^4}{32/T + 8\Delta^2/3}}\\
&\le 2 e^{-256 (T |\Bhat_i|\,|\Bhat_j| \sqrt{L} \Delta)} + 64 n^4 e^{-\frac{S\Delta^4}{32/T + 8\Delta^2/3}}.
\end{align*}
\end{proof}

\begin{lemma}\label{eq:bound over E}
Let $E$ be a subset of the edge set $E_0 = \{(i,j) \;|\; i \in \{1,\ldots,n\}, j \in \{1,\ldots,n\} \}$. Then under the above setup, there exists constants $c_0$ and $c_1$ such that
\begin{equation}
\Pr\left[ \frac{1}{|E|} \sum_{i_v, j_v \in E} |w(u_{i_v},u_{j_v}) - \what_{ij}| > c_0 \sqrt{\Delta} \right]
\le \sum_{i_v, j_v \in E} 2 e^{-c_1 (T |\Bhat_i|\,|\Bhat_j| \Delta)} + 64 |E| n^4 e^{-\frac{S\Delta^4}{32/T + 8\Delta^2/3}}.
\end{equation}
\end{lemma}

\begin{proof}
From Lemma \ref{lemma:bound w and what}, average over all pairs $(i_v,j_v) \in E$,
\begin{align*}
\Pr\left[ \frac{1}{|E|} \sum_{i_v, j_v \in E} |w(u_{i_v},u_{j_v}) - \what_{ij}| > 16\Delta^{1/2} L^{1/4} \right]
&\le \frac{1}{|E|} \sum_{i_v, j_v \in E} \Pr\left[  |w(u_{i_v},u_{j_v}) - \what_{ij} > 16\Delta^{1/2} L^{1/4} |\right] \\
&\le \sum_{i_v, j_v \in E} 2 e^{-256 (T |\Bhat_i|\,|\Bhat_j| \sqrt{L} \Delta)} + 64 |E| n^4 e^{-\frac{S\Delta^4}{32/T + 8\Delta^2/3}}.
\end{align*}
Choosing $c_0 = 16L^{1/4}$ and $c_1 = 256\sqrt{L}$ yields the desired result.
\end{proof}

\begin{lemma}
\label{lemma:consistency1}
Let $I_k = [\alpha_{k-1},\alpha_k]$ for $k = 1, \ldots,K$ be a sequence of intervals such that $I_i \cap I_j = \emptyset$ and $\cup I_i = [0,1]$. Suppose that $w$ is piecewise Lipschitz continuous and differentiable in $I_k$. For any $u_i,u_j \in [0,1]$, define
\begin{align*}
f_{ij}(x) &= \left(w(x,u_i) - w(x,u_j)\right)^2\\
g_{ij}(y) &= \left(w(u_i,y) - w(u_j,y)\right)^2,
\end{align*}
and
\begin{align*}
h_{ij}(x,y) = \frac{1}{2}\left[ f_{ij}(x) + g_{ij}(y) \right].
\end{align*}

Let $\delta = \min\limits_{k = 1,\ldots,K} |\alpha_k - \alpha_{k-1}|$. If
\begin{align*}
d_{ij}^c = \int_0^1 f_{ij}(x) dx \le \frac{\epsilon^2}{8L}, \quad\quad \mbox{and}\quad\quad d_{ij}^r = \int_0^1 g_{ij}(y) dy \le \frac{\epsilon^2}{8L},
\end{align*}
for some constant $0 < \epsilon < 2\delta L$, then
\begin{align*}
\sup_{x \in [0,1]} f_{ij}(x) \le \epsilon, \quad\quad \mbox{and}\quad\quad \sup_{y \in [0,1]} g_{ij}(y) \le \epsilon.
\end{align*}
Hence, $\sup\limits_{(x,y) \in [0,1]^2} h_{ij}(x,y) \le \epsilon$.
\end{lemma}

\begin{proof}
Since $h_{ij}(x,y)$ is separable, it is sufficient to prove for $f_{ij}(x)$.

Fix $i$ and $j$, and let $f_{ij}^{sup} = \sup\limits_{x \in [0,1]} f_{ij}(x)$. Let $I_k = [\alpha_{k-1},\;\alpha_k]$ be the interval such that $f_{ij}^{sup}$ is attained, and let $\delta_k = |\alpha_k - \alpha_{k-1}|$ be the width of the interval. Consider a neighborhood surrounding the center of $I_k$ with radius $\delta_k/2 - \theta$, where $0 < \theta < \delta_k/2$. Then define
\begin{align*}
f_{ij}^{sup}(\theta) = \sup\limits_{x \in [\alpha_{k-1}+\theta, \alpha_k-\theta]} f_{ij}(x).
\end{align*}
It is clear that $f_{ij}^{sup} = \lim\limits_{\theta \rightarrow 0} f_{ij}^{sup}(\theta)$.

The set $[\alpha_{k-1}+\theta, \alpha_k-\theta]$ is compact, so there exists $x^{max}_{ij}(\theta) \in [\alpha_{k-1}+\theta, \alpha_k-\theta]$ such that $f_{ij}^{sup} = f_{ij}(x^{max}_{ij})$. Assume, without lost of generality, that $x^{max}_{ij}(\theta) + \delta_k/2 - \theta$ (\emph{i.e.}, $x^{max}_{ij}$ is in the lower half of the interval). For any $0 < \epsilon_0 < \frac{\epsilon}{4L} - \theta \le \frac{\delta}{2}-\theta \le \frac{\delta_k}{2} - \theta$,
$$\frac{h_{ij}(x_{ij}^{max}(\theta)) - h_{ij}(x_{ij}^{max}(\theta)+\epsilon_0)}{\epsilon_0}= $$
$$\frac{(w(i,x_{ij}^{max})-w(j,x_{ij}^{max}))^2-(w(i,x_{ij}^{max}(\theta)+\epsilon_0)-w(j,x_{ij}^{max}(\theta)+\epsilon_0))^2}{\epsilon_0}\leq$$ 
$$ \frac{(w(i,x_{ij}^{max})-w(j,x_{ij}^{max}))^2-(w(i,x_{ij}^{max})+L\epsilon_0-w(j,x_{ij}^{max})+L\epsilon_0)^2}{\epsilon_0}\leq$$
$$ {4L (w(j,x_{ij}^{max})-w(i,x_{ij}^{max}))}\leq 4L\Rightarrow$$

\begin{align*}
\frac{f_{ij}(x_{ij}^{max}(\theta)) - f_{ij}(x_{ij}^{max}(\theta) + \epsilon_0)}{\epsilon_0} \le 4L,
\end{align*}
which implies that
\begin{align*}
f_{ij}(x_{ij}^{max}(\theta)) - 4L \epsilon_0 \le f_{ij}(x_{ij}^{max}(\theta)+\epsilon_0).
\end{align*}
Integrating both sides with respect to $\epsilon_0$ with limits $0$ and $\frac{\epsilon}{4L}-\theta$ yields
\begin{align*}
f_{ij}(x_{ij}^{max}(\theta))\left(\frac{\epsilon}{4L} - \theta\right) - \frac{4L}{2}\left(\frac{\epsilon}{4L}-\theta\right)^2
&\le \int_0^{\frac{\epsilon}{4L}-\theta} f_{ij}(x_{ij}^{max}(\theta)+\epsilon_0) d\epsilon_0 \\
&\le \int_0^1 f_{ij}(x) dx = d_{ij}^c.
\end{align*}
Therefore,
\begin{align*}
f_{ij}(x_{ij}^{max}(\theta)) \le \frac{d_{ij}^c}{\frac{\epsilon}{4L}-\theta} + 2L \left(\frac{\epsilon}{4L} - \theta\right),
\end{align*}
and hence
\begin{align*}
f_{ij}^{sup} = \lim_{\theta \rightarrow 0} f_{ij}^{sup}(\theta) = \lim_{\theta \rightarrow 0} f_{ij}(x_{ij}^{max}(\theta)) \le \frac{4Ld_{ij}^c}{\epsilon} + \frac{\epsilon}{2}.
\end{align*}
It then follows that if $d_{ij}^c \le \frac{\epsilon^2}{8L}$, then $f_{ij}^{sup} \le \epsilon$.
\end{proof}

\begin{definition}
\label{def:MSE and MAE}
The mean squared error (MSE) and mean absolute error (MAE) are defined as
\begin{align}
\MSE(\what) &= \frac{1}{n^2} \sum_{i_v = 1}^n \sum_{j_v = 1}^n \left( w(u_{i_v}, u_{j_v}) - \what_{i_v, j_v} \right)^2\\
\MAE(\what) &= \frac{1}{n^2} \sum_{i_v = 1}^n \sum_{j_v = 1}^n \left| w(u_{i_v}, u_{j_v}) - \what_{i_v, j_v} \right|.
\end{align}
\end{definition}

\begin{theorem}
If $S \in \Theta(n)$ and $\Delta_n \in \omega \left( \left(\frac{\log(n)}{n}\right)^{\frac{1}{4}} \right) \cap o(1)$, then
\begin{equation}
\lim_{n\rightarrow \infty} \E[ \MAE(\what)] = 0 \quad\quad\mbox{and}\quad\quad \lim_{n\rightarrow \infty} \E[ \MSE(\what)] = 0.
\end{equation}
\label{thm:consistency}
\end{theorem}
\begin{proof}
Suppose that the algorithm is executed for a set of observed graphs with $n$ vertices using parameters $\Delta_n$ and $S$. Let $K'_n$ be the number of blocks generated. Assume that, as $n \rightarrow \infty$, the parameters satisfy $S \in \Theta(n)$ and $\Delta_n \in \omega \left( \left(\frac{\log(n)}{n}\right)^{\frac{1}{4}} \right) \cap o(1)$.

The proof is based on \eref{eq:bound over E}. The intuition is to that that the two terms $\sum_{i_v, j_v \in E} 2 e^{-c_1 (T |\Bhat_i|\,|\Bhat_j| \Delta)}$ and $32 |E| n^4 e^{-\frac{S\Delta^4}{16/T + 8\Delta^2/3}}$ vanish as $n \rightarrow \infty$. The latter is clear if $S \in \Theta(n)$ and $\Delta_n \in \omega \left( \left(\frac{\log(n)}{n}\right)^{\frac{1}{4}} \right) \cap o(1)$. For the first term, it is necessary to consider the size $|E|$, which is the size of the cluster generated. We show that the number of small clusters is asymptotically irrelevant. Most of the error come from vertices whose cluster is large enough to make $e^{-\frac{S\Delta^4}{32/T + 8\Delta^2/3}}$ vanish.

From Theorem \ref{thm:number of blocks}, we have
\begin{equation*}
\Pr\left[ K' > \frac{QL\sqrt{2}}{\Delta_n} \right] \le 8n^2 e^{-\frac{S\Delta_n^4}{128/T + 16\Delta_n^2/3}}.
\end{equation*}
Let $\calE_n$ be the event that $K_n' \le QL\sqrt{2}/\Delta_n$. Then $\lim_{n\rightarrow \infty} \Pr[\calE_n] = 1$.

Suppose $\calE_n$ happens and define $r_n$ as the number of blocks with less than $\frac{n\Delta_n^2}{QL\sqrt{2}}$ elements. Let $V_n$ be the union of these blocks, and define $\Vbar_n$ be the complement of $V_n$. Then,
\begin{align*}
|V_n| \le r_n \frac{n\Delta_n^2}{QL\sqrt{2}} \le K_n' \frac{n\Delta_n^2}{QL\sqrt{2}} \le n\Delta_n.
\end{align*}
So, $|V_n|/n \le \Delta_n$.

Now, let's consider MAE.
\begin{align*}
\MAE &= \frac{1}{n^2} \sum_{i_v \in V} \sum_{j_v \in V} \left| w(u_{i_v}, u_{j_v}) - \what_{i_v, j_v} \right|\\
&= \frac{1}{n^2} \sum_{i_v \in V_n} \sum_{j_v \in V_n} \left| w(u_{i_v}, u_{j_v}) - \what_{i_v, j_v} \right| +
\frac{1}{n^2}\sum_{i_v \in \Vbar_n} \sum_{j_v \in \Vbar_n} \left| w(u_{i_v}, u_{j_v}) - \what_{i_v, j_v} \right| + \\
&\quad + \frac{1}{n^2}\sum_{i_v \in \Vbar_n} \sum_{j_v \in V_n} \left| w(u_{i_v}, u_{j_v}) - \what_{i_v, j_v} \right| +
\frac{1}{n^2}\sum_{i_v \in V_n} \sum_{j_v \in \Vbar_n} \left| w(u_{i_v}, u_{j_v}) - \what_{i_v, j_v} \right|\\
&\le \frac{|V_n|^2}{n^2} + \frac{|V_n|}{n}\frac{|\Vbar_n|}{n} + \frac{|\Vbar_n|}{n}\frac{|V_n|}{n}
+ \frac{1}{n^2}\sum_{i_v \in \Vbar_n} \sum_{j_v \in \Vbar_n} \left| w(u_{i_v}, u_{j_v}) - \what_{i_v, j_v} \right|\\
&\le \frac{1}{n^2}\sum_{i_v \in \Vbar_n} \sum_{j_v \in \Vbar_n} \left| w(u_{i_v}, u_{j_v}) - \what_{i_v, j_v} \right| + \Delta_n^2 + 2\Delta_n\\
&\le \frac{1}{n^2}\sum_{i_v \in \Vbar_n} \sum_{j_v \in \Vbar_n} \left| w(u_{i_v}, u_{j_v}) - \what_{i_v, j_v} \right| + 3\Delta_n.
\end{align*}
Similar result holds for MSE:
\begin{align*}
\MSE &= \frac{1}{n^2} \sum_{i_v \in V} \sum_{j_v \in V} \left( w(u_{i_v}, u_{j_v}) - \what_{i_v, j_v} \right)^2
\le \frac{1}{n^2}\sum_{i_v \in \Vbar_n} \sum_{j_v \in \Vbar_n} \left( w(u_{i_v}, u_{j_v}) - \what_{i_v, j_v} \right)^2 + 3\Delta_n.
\end{align*}

Therefore, using Lemma \ref{eq:bound over E} with $E = \Vbar_n$:
\begin{align*}
\Pr\left[ \MAE(\what) > c_0\sqrt{\Delta_n} + 3\Delta_n \;\Big|\; \calE \right]
&\le \Pr\left[ \frac{1}{n^2}\sum_{i_v \in \Vbar_n} \sum_{j_v \in \Vbar_n} \left| w(u_{i_v}, u_{j_v}) - \what_{i_v, j_v} \right| + 3\Delta_n > c_0\sqrt{\Delta_n} + 3\Delta_n \;\Big|\; \calE \right]\\
&\le \frac{1}{\Pr[\calE]} \Pr\left[ \frac{1}{|\Vbar_n|^2}\sum_{i_v \in \Vbar_n} \sum_{j_v \in \Vbar_n} \left| w(u_{i_v}, u_{j_v}) - \what_{i_v, j_v} \right|  > c_0\sqrt{\Delta_n} \;\Big|\; \calE \right]\\
&\le \frac{1}{\Pr[\calE]} \left( \sum_{i_v \in \Vbar_n}\sum_{j_v \in \Vbar_n} 2 e^{-256 (T |\Bhat_i|\,|\Bhat_j| \sqrt{L} \Delta)} + 64 |\Vbar_n| n^4 e^{-\frac{S\Delta^4}{32/T + 8\Delta^2/3}} \right).
\end{align*}
and
\begin{align*}
\Pr\left[ \MSE(\what) > c_0\sqrt{\Delta_n} + 3\Delta_n \;\Big|\; \calE \right]
&\le \frac{1}{\Pr[\calE]} \left( \sum_{i_v \in \Vbar_n}\sum_{j_v \in \Vbar_n} 2 e^{-256 (T |\Bhat_i|\,|\Bhat_j| \sqrt{L}\Delta)} + 64 |\Vbar_n| n^4 e^{-\frac{S\Delta^2}{32/T + 8\Delta/3}} \right).
\end{align*}

So,
\begin{align*}
\lim_{n \rightarrow \infty} \Pr\left[ \MAE(\what) > c_0\sqrt{\Delta_n} + 3\Delta_n \;\Big|\; \calE \right] \Pr\left[\calE\right] = 0.
\end{align*}
Since $\lim_{n \rightarrow \infty} \Delta_n = 0$ and $\lim_{n \rightarrow \infty} \Pr[\calE_n] = 1$, it holds that for any $\epsilon > 0$,
\begin{align*}
\lim_{n\rightarrow \infty} \Pr[\MAE(\what) > \epsilon ] = 0.
\end{align*}

Finally, since $\what$ is bounded in $[0,1]$,
\begin{align*}
\E[\MAE(\what)] \le \epsilon \Pr[\MAE(\what) \le \epsilon] + \Pr[\MAE(\what) > \epsilon].
\end{align*}
Sending $\epsilon \rightarrow \infty$,
\begin{align*}
\lim_{n\rightarrow \infty} \E[\MAE(\what)] \le \lim_{n\rightarrow \infty} \Pr[\MAE(\what) > \epsilon] = 0.
\end{align*}
Same arguments hold for $\MSE$.
\end{proof}

\end{document}